\documentclass[pdflatex,sn-mathphys-num]{sn-jnl}

\usepackage{graphicx}%
\usepackage{multirow}%
\usepackage{amsmath,amssymb,amsfonts}%
\usepackage{amsthm}%
\usepackage{mathrsfs}%
\usepackage[title]{appendix}%
\usepackage{xcolor}%
\usepackage{textcomp}%
\usepackage{manyfoot}%
\usepackage{booktabs}%
\usepackage{algorithm}%
\usepackage{algorithmicx}%
\usepackage{algpseudocode}%
\usepackage{listings}%
\usepackage{booktabs}
\usepackage{caption}
\usepackage{subfigure}
\usepackage[table]{xcolor}
\usepackage{colortbl}
\definecolor{cellgray}{RGB}{220, 220, 220}   
\definecolor{bestfg}{RGB}{210, 90, 100}      
\definecolor{secondfg}{RGB}{60, 160, 130}    
\usepackage[table]{xcolor}
\theoremstyle{thmstyleone}%
\newtheorem{theorem}{Theorem}

\theoremstyle{thmstyletwo}%
\newtheorem{remark}{Remark}%
\theoremstyle{thmstylethree}%

\begin{document}

\title{Physics-informed line-of-sight learning for scalable deterministic channel modeling}

\author[1]{\fnm{Xiucheng} \sur{Wang}}\email{xcwang\_1@stu.xidian.edu.cn}

\author[1]{\fnm{Junxi} \sur{Huang}}\email{24012100067@stu.xidian.edu.cn}

\author*[1]{\fnm{Conghao} \sur{Zhou}}\email{conghao.zhou@ieee.org}

\author[2]{\fnm{Xuemin} \sur{Shen}}\email{sshen@uwaterloo.ca}

\author*[1]{\fnm{Nan} \sur{Cheng}}\email{dr.nan.cheng@ieee.org}

\affil*[1]{\orgdiv{School of Telecommunications Engineering}, \orgname{Xidian University}, \orgaddress{\street{No. 2 South Taibai Road}, \city{Xi'an}, \postcode{710071}, \state{Shaanxi}, \country{China}}}

\affil[2]{\orgdiv{Department of Electrical and Computer Engineering}, \orgname{University of Waterloo}, \orgaddress{\street{200 University Avenue West}, \city{Waterloo}, \postcode{N2L 3G1}, \state{Ontario}, \country{Canada}}}

\abstract{Deterministic channel modeling maps a physical environment to its site-specific electromagnetic response. Ray tracing produces complete multi-dimensional channel information but remains prohibitively expensive for area-wide deployment. We identify line-of-sight (LoS) region determination as the dominant bottleneck. To address this, we propose D$^2$LoS, a physics-informed neural network that reformulates dense pixel-level LoS prediction into sparse vertex-level visibility classification and projection point regression, avoiding the spectral bias at sharp boundaries. A geometric post-processing step enforces hard physical constraints, yielding exact piecewise-linear boundaries. Because LoS computation depends only on building geometry, cross-band channel information is obtained by updating material parameters without retraining. We also construct RayVerse-100, a ray-level dataset spanning 100 urban scenarios with per-ray complex gain, angle, delay, and geometric trajectory. Evaluated against rigorous ray tracing ground truth, D$^2$LoS achieves 3.28~dB mean absolute error in received power, 4.65$^\circ$ angular spread error, and 20.64~ns delay spread error, while accelerating visibility computation by over 25$\times$.}


\maketitle

\section{Introduction}

Deterministic channel modeling maps a physical environment to its exact electromagnetic response. It provides site-specific channel state information (CSI) for 6G network optimization \cite{zheng2025radio}, digital twin construction \cite{11152929,11278649}, and physical-layer AI applications \cite{zeng2024tutorial,wang2024tutorial,liu2025wifo}. Unlike stochastic models that discard spatial structure by design \cite{eyceoz2002deterministic,khuwaja2018survey,kammoun2014preliminary}, deterministic approaches preserve the gain, angle, delay, and geometric trajectory of every propagation path \cite{deschamps1972ray}. This level of detail is essential for coherent beamforming \cite{tang2021optimization,ha2024radio}, MIMO channel matrix construction \cite{oh2004mimo}, and beam management in dense urban deployments \cite{chen2025learning,hoffmann2023beam}. However, generating such complete channel information at area-wide scale remains an open challenge. Ray tracing (RT) is the only method that produces multi-dimensional channel parameters in a single computation \cite{deschamps1972ray,hoydis2023sionna}, yet its cost is prohibitive for large-scale deployment \cite{he2018design}. This tension between information completeness and computational scalability motivates the search for efficient alternatives.

Existing learning-based approaches address only fragments of this problem \cite{wang2024radiodiff,zhang2023rme,wang2025iradiodiff,wang2026tutorial}. Link-level methods predict point-to-point channel parameters from local features but cannot produce spatial coverage maps \cite{guo2022overview}. Area-level radio map (RM) construction methods provide coverage predictions, yet most focus on received power alone \cite{levie2021radiounet}. Recent efforts have explored angular and delay profile estimation, but each dimension requires a separately trained model at considerable cost. Meanwhile, available datasets offer either aggregate path-loss maps or dominant-path summaries \cite{11083758,wang2026tutorial}. A few datasets provide per-path angle and delay parameters, but explicit geometric ray trajectories remain absent \cite{11083758,10122907}. This gap between what AI-driven network optimization demands and what current methods deliver remains unresolved.

We make two key observations that guide our approach. First, the computational bottleneck of RT lies not in the multi-bounce path search but in the preceding line-of-sight (LoS) preprocessing step \cite{600437,10069838}. For a scenario with $m$ building vertices, the classical rotational sweep algorithm requires $O(m \log m)$ operations per transmitter \cite{choi2023withray}. The subsequent path search operates on precomputed visibility adjacency via bounded-depth graph traversal, whose per-query cost is substantially lower \cite{valenzuela1993ray,yang1998ray}. Second, in a 2D environment of polygonal obstacles, the LoS boundary seen from any point source is piecewise-linear. Each boundary segment is defined by a visible building vertex and its projection onto a farther building edge. These two observations suggest a natural decomposition: the network predicts sparse vertex-level attributes, and the sharp boundary is reconstructed analytically.

Here we introduce D$^2$LoS\footnote{Code is available at \url{https://github.com/UNIC-Lab/D2LoS}.}, a physics-informed neural network that addresses the LoS bottleneck through vertex-level geometric decomposition. We also release RayVerse-100\footnote{Dataset is available at \url{https://github.com/UNIC-Lab/RayVerse}.}, a ray-level dataset spanning 100 diverse urban scenarios with per-ray complex gain, angle, delay, and complete geometric trajectory. The main contributions are as follows.
\begin{enumerate}
    \item We reformulate LoS map prediction from dense pixel-level classification into sparse vertex-level visibility classification and projection point regression. This formulation avoids the spectral bias of neural networks at sharp LoS boundaries. Combined with a geometric post-processing step that enforces hard physical constraints, D$^2$LoS reduces per-transmitter preprocessing complexity from $O(m \log m)$ to $O(M \log M + n)$, where $M$ is the number of boundary vertices with $M \ll m$ and $n$ is the number of evaluation points. In typical urban scenarios, $M \log M \ll n$, yielding an effective cost of $O(n)$.

    \item Because LoS computation depends only on building geometry, it is inherently frequency-independent. Once geometric ray paths are determined, cross-band channel information is obtained by updating material parameters through uniform theory of diffraction (UTD) \cite{tsingos2001modeling} formulations without retraining. This enables unified multi-dimensional radio map construction covering received signal strength (RSS), angular power spectrum (APS), and power delay profile (PDP) from a single pipeline.

    \item We construct RayVerse-100, a ray-level dataset spanning 100 diverse urban scenarios. Each record captures per-ray complex gain, angle of arrival (AoA), angle of departure (AoD), propagation delay, and complete geometric trajectory with 3D reflection-point coordinates. To our knowledge, this is the first open dataset providing explicit per-ray geometric trajectories at this scale. We also provide post-processing interfaces for coherent superposition, MIMO channel matrix synthesis, and cross-band CSI generation.

    \item Evaluated against rigorous RT ground truth on unseen test scenarios, D$^2$LoS achieves 3.28~dB mean absolute error in RSS, 4.65$^\circ$ angular spread error, and 20.64~ns delay spread error, while accelerating visibility computation by over 25$\times$.
\end{enumerate}

\section{Results}

We evaluate D$^2$LoS on the RayVerse-100 dataset against three baselines. No-Geom removes the geometric post-processing from D$^2$LoS and uses raw neural network predictions directly. RadioUNet is a U-Net architecture that predicts pixel-level LoS maps. RMTransformer is a vision transformer variant designed for radio map estimation. All methods receive the same input and are trained on the same LoS corpus. The ground truth is produced by exact rotational sweep ray tracing with UTD field computation. For each evaluation point, the top-8 rays ranked by path gain are retained.

\subsection*{D$^2$LoS achieves high-fidelity multi-dimensional radio maps}

\begin{table}[ht]
\centering
\captionsetup{font={small}, skip=16pt}
\caption{\textbf{Received signal strength prediction accuracy across 100 test scenarios.} Six metrics are reported as mean $\pm$ standard deviation over all scenarios in RayVerse-100. D$^2$LoS achieves the lowest error and highest correlation across all metrics. No-Geom denotes D$^2$LoS without geometric post-processing. RadioUNet and RMTransformer are pixel-level LoS prediction baselines. Bold values indicate the best performance.}
\label{tab:rss_comparison}
\begin{tabular}{lcccc}
\toprule
Metric & $D^2$LoS & No-Geom & RadioUNet & RMTransformer \\
\midrule
Bias (dB) & \cellcolor{cellgray!40}\textcolor{bestfg}{\textbf{-1.69 $\pm$ 3.82}} & \textcolor{secondfg}{-2.96 $\pm$ 11.76} & -6.78 $\pm$ 26.30 & -26.60 $\pm$ 28.39 \\
MAE (dB) & \cellcolor{cellgray!40}\textcolor{bestfg}{\textbf{3.28 $\pm$ 3.82}} & \textcolor{secondfg}{12.65 $\pm$ 10.01} & 32.80 $\pm$ 11.87 & 40.85 $\pm$ 15.98 \\
RMSE (dB) & \cellcolor{cellgray!40}\textcolor{bestfg}{\textbf{7.28 $\pm$ 4.57}} & \textcolor{secondfg}{19.92 $\pm$ 11.27} & 40.12 $\pm$ 11.38 & 48.11 $\pm$ 15.42 \\
MSE (dB$^2$) & \cellcolor{cellgray!40}\textcolor{bestfg}{\textbf{73.93 $\pm$ 174.32}} & \textcolor{secondfg}{523.75 $\pm$ 725.86} & 1738.76 $\pm$ 1120.59 & 2552.54 $\pm$ 1616.15 \\
NMSE & \cellcolor{cellgray!40}\textcolor{bestfg}{\textbf{0.0088 $\pm$ 0.0297}} & \textcolor{secondfg}{0.0826 $\pm$ 0.1481} & 0.2695 $\pm$ 0.2756 & 0.4789 $\pm$ 0.4912 \\
Correlation & \cellcolor{cellgray!40}\textcolor{bestfg}{\textbf{0.9537 $\pm$ 0.0603}} & \textcolor{secondfg}{0.6473 $\pm$ 0.3076} & 0.1350 $\pm$ 0.2824 & 0.1708 $\pm$ 0.3084 \\
\bottomrule
\end{tabular}
\end{table}

\begin{table}[ht]
\centering
\captionsetup{font={small}, skip=16pt}
\caption{\textbf{Angular power spectrum prediction accuracy.} Four metrics are reported as the mean over all evaluation points across 100 test scenarios. Angular spread (AS) absolute error and mean direction of arrival (MDoA) absolute error measure first-order angular statistics. Shape cosine and shape RMSE quantify the similarity between predicted and ground-truth APS profiles.}
\label{tab:aps_comparison}
\begin{tabular}{lcccc}
\toprule
Metric & $D^2$LoS & No-Geom & RadioUNet & RMTransformer \\
\midrule
AS Abs.\ Err.\ (deg) & \cellcolor{cellgray!40}\textcolor{bestfg}{\textbf{4.65}} & \textcolor{secondfg}{11.80} & 22.91 & 23.37 \\
MDoA Abs.\ Err.\ (deg) & \cellcolor{cellgray!40}\textcolor{bestfg}{\textbf{17.07}} & \textcolor{secondfg}{37.72} & 70.09 & 74.64 \\
Shape Cosine & \cellcolor{cellgray!40}\textcolor{bestfg}{\textbf{0.8924}} & \textcolor{secondfg}{0.6667} & 0.3158 & 0.2902 \\
Shape RMSE & \cellcolor{cellgray!40}\textcolor{bestfg}{\textbf{0.0035}} & \textcolor{secondfg}{0.0090} & 0.0170 & 0.0177 \\
\bottomrule
\end{tabular}
\end{table}

\begin{table}[ht]
\centering
\captionsetup{font={small}, skip=16pt}
\caption{\textbf{Power delay profile prediction accuracy.} Six metrics are reported as the mean over all evaluation points across 100 test scenarios. Delay spread (DS) and median delay capture temporal dispersion. $K$-factor error reflects the dominant-to-scattered path ratio accuracy. Effective count measures the number of resolvable multipath components.}
\label{tab:pdp_comparison}
\begin{tabular}{lcccc}
\toprule
Metric & $D^2$LoS & No-Geom & RadioUNet & RMTransformer \\
\midrule
DS Abs.\ Err.\ (ns) & \cellcolor{cellgray!40}\textcolor{bestfg}{\textbf{20.64}} & \textcolor{secondfg}{46.11} & 90.93 & 97.94 \\
Median Delay Abs.\ Err.\ (ns) & \cellcolor{cellgray!40}\textcolor{bestfg}{\textbf{47.91}} & \textcolor{secondfg}{126.67} & 245.72 & 237.56 \\
$K$-factor Abs.\ Err.\ (dB) & \cellcolor{cellgray!40}\textcolor{bestfg}{\textbf{2.23}} & \textcolor{secondfg}{4.61} & 8.77 & 9.41 \\
Effective Count Abs.\ Err. & \cellcolor{cellgray!40}\textcolor{bestfg}{\textbf{4.88}} & \textcolor{secondfg}{11.82} & 26.78 & 26.48 \\
Shape Cosine & \cellcolor{cellgray!40}\textcolor{bestfg}{\textbf{0.6434}} & \textcolor{secondfg}{0.4610} & 0.1743 & 0.1407 \\
Shape RMSE & \cellcolor{cellgray!40}\textcolor{bestfg}{\textbf{0.0327}} & \textcolor{secondfg}{0.0496} & 0.0743 & 0.0763 \\
\bottomrule
\end{tabular}
\end{table}

D$^2$LoS consistently outperforms all baselines across three channel dimensions. In received signal strength (RSS) prediction, D$^2$LoS achieves a mean absolute error (MAE) of 3.28~dB and a Pearson correlation of 0.954 averaged over 100 test scenarios (Table~\ref{tab:rss_comparison}). The closest baseline, No-Geom, yields an MAE of 12.65~dB and a correlation of 0.647. RadioUNet and RMTransformer produce MAEs exceeding 32~dB with correlations below 0.18, indicating near-complete failure in spatial power prediction. The advantage extends to angular and delay domains. For the angular power spectrum (APS), D$^2$LoS achieves an angular spread error of 4.65$^\circ$ and a shape cosine similarity of 0.892 (Table~\ref{tab:aps_comparison}). For the power delay profile (PDP), D$^2$LoS achieves a delay spread error of 20.64~ns and a $K$-factor error of 2.23~dB (Table~\ref{tab:pdp_comparison}). In both domains, RadioUNet and RMTransformer yield shape cosine values below 0.32. These low values indicate that their predicted angular and delay profiles bear little resemblance to the ground truth.

\begin{figure}[ht]
  \centering
  \includegraphics[width=\linewidth, height=0.38\textheight, keepaspectratio]{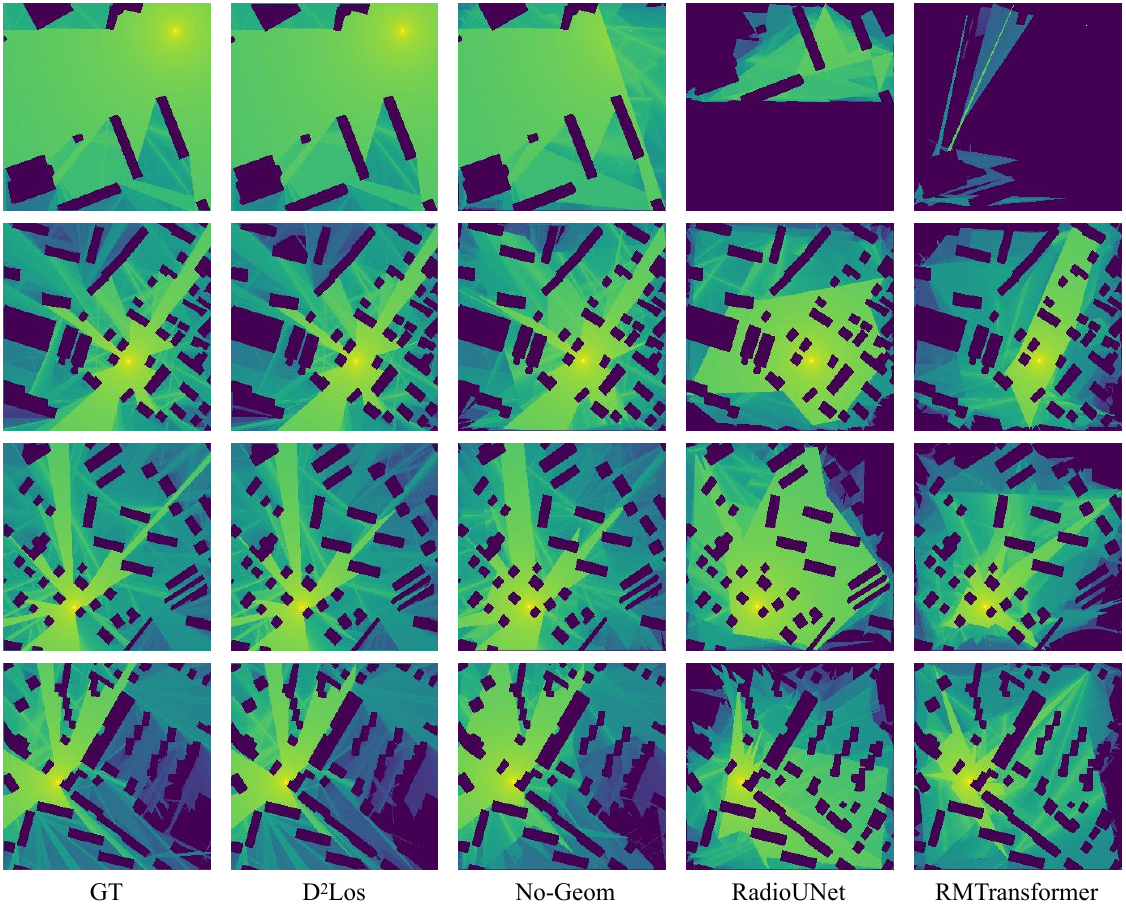}
  \caption{\textbf{Qualitative comparison of RSS radio maps across four representative scenarios.} Each row corresponds to one scenario from the RayVerse-100 test set. Columns from left to right: ground truth (GT), D$^2$LoS, No-Geom, RadioUNet \cite{levie2021radiounet}, and RMTransformer \cite{li2025rmtransformer}. D$^2$LoS reproduces the spatial power distribution with high fidelity. No-Geom captures the coarse pattern but exhibits blurred boundaries. RadioUNet and RMTransformer fail to reconstruct meaningful spatial structures.}
  \label{fig:rss}
\end{figure}

Fig.~\ref{fig:rss} provides a qualitative comparison across four representative scenarios. D$^2$LoS reproduces both the LoS regions near the transmitter and the shadow boundaries behind buildings. No-Geom captures the coarse spatial pattern but produces blurred transitions at building edges. RadioUNet and RMTransformer fail to reconstruct meaningful spatial structures. RadioUNet generates near-uniform power distributions that lose all shadow boundary detail. RMTransformer produces narrow beam-like artifacts that do not correspond to any physical propagation pattern. These qualitative observations are consistent with the quantitative metrics in Tables~\ref{tab:rss_comparison}--\ref{tab:pdp_comparison}.

\subsection*{Geometric post-processing is critical for boundary-sensitive metrics}

To isolate the contribution of geometric post-processing, we compare D$^2$LoS with its ablated variant No-Geom. This variant uses raw neural network predictions without edge snapping or ray-edge intersection refinement. The impact of removing this step is substantial across all three channel dimensions. In RSS, the MAE increases from 3.28 to 12.65~dB and the correlation drops from 0.954 to 0.647 (Table~\ref{tab:rss_comparison}). In the angular domain, the APS shape cosine decreases from 0.892 to 0.667, and the mean direction of arrival error grows from 17.07$^\circ$ to 37.72$^\circ$ (Table~\ref{tab:aps_comparison}). In the delay domain, the degradation is even more pronounced. The PDP shape cosine drops from 0.643 to 0.461, and the median delay error increases from 47.91 to 126.67~ns (Table~\ref{tab:pdp_comparison}). These results confirm that geometric post-processing is not a marginal refinement but a critical component of the pipeline.

\begin{figure}[ht]
  \centering
  \includegraphics[width=\linewidth, height=0.38\textheight, keepaspectratio]{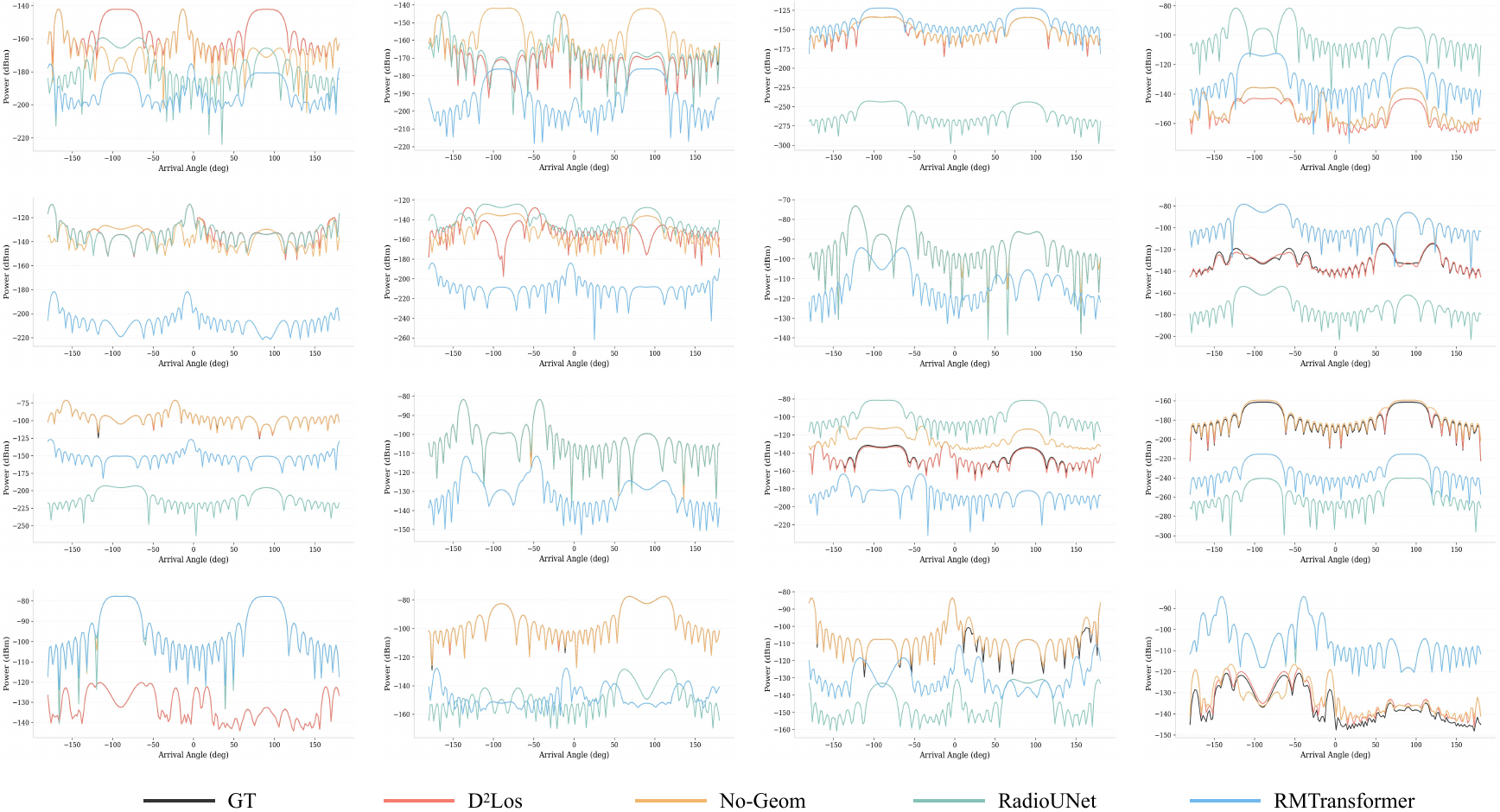}
  \caption{\textbf{Qualitative comparison of angular power spectra at selected receiver locations.} Each row corresponds to one receiver position from a different test scenario. The horizontal axis is the azimuth arrival angle and the vertical axis is the received power in dBm. D$^2$LoS closely tracks the ground truth. No-Geom preserves dominant arrival directions but introduces errors in secondary peaks. RadioUNet and RMTransformer produce noisy angular profiles with incorrect peak positions.}
  \label{fig:aps}
\end{figure}

\begin{figure}[ht]
  \centering
  \includegraphics[width=\linewidth, height=0.38\textheight, keepaspectratio]{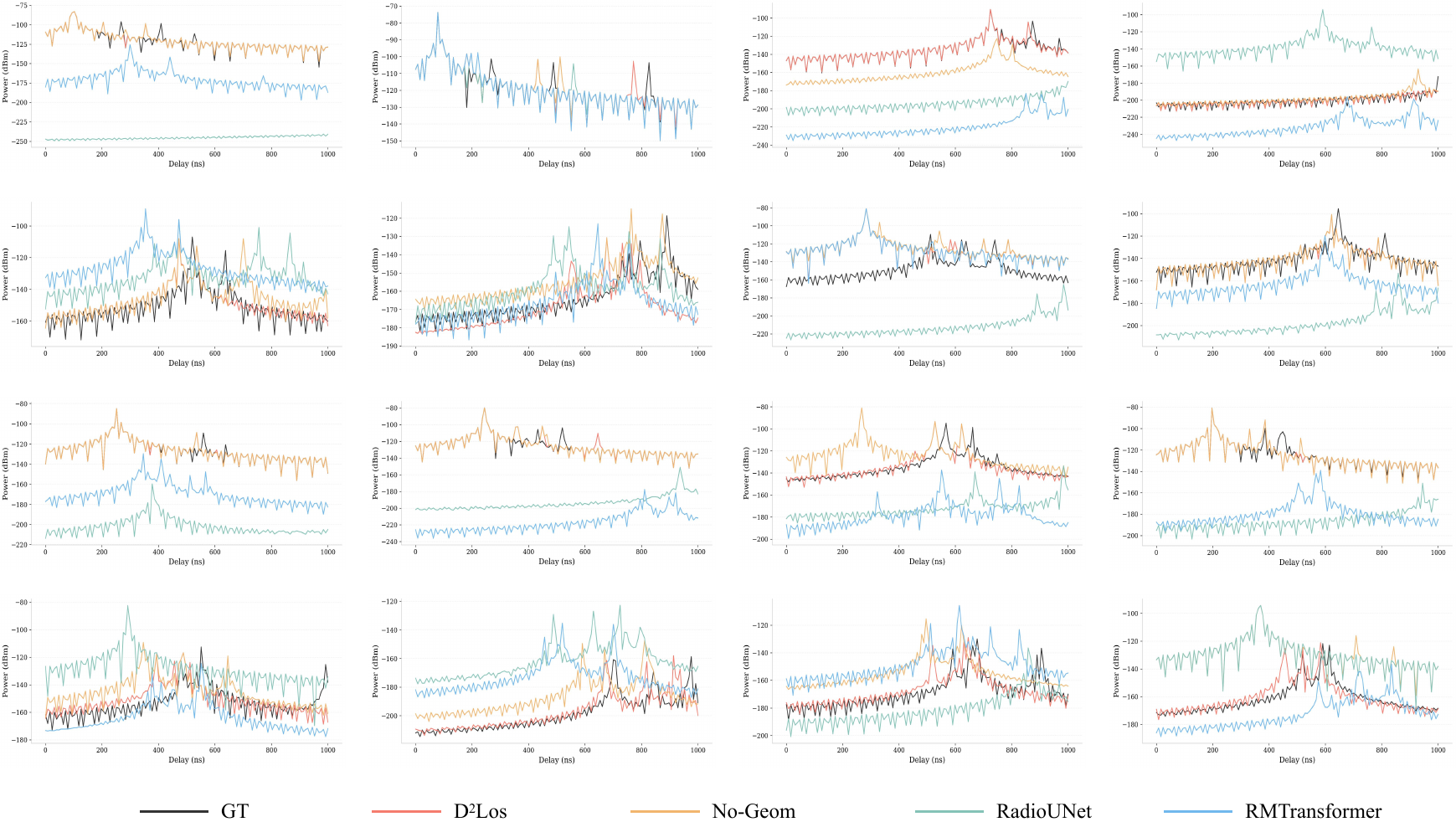}
  \caption{\textbf{Qualitative comparison of power delay profiles at selected receiver locations.} Each row corresponds to one receiver position from a different test scenario. The horizontal axis is the propagation delay in nanoseconds and the vertical axis is the received power in dBm. D$^2$LoS accurately reproduces the delay-domain structure. No-Geom exhibits shifted peak positions. RadioUNet and RMTransformer produce delay profiles with spurious peaks at large delays.}
  \label{fig:pdp}
\end{figure}

The qualitative comparisons reinforce these findings. In the APS curves (Fig.~\ref{fig:aps}), No-Geom preserves the dominant arrival directions but introduces noticeable errors in secondary peaks. In the PDP curves (Fig.~\ref{fig:pdp}), No-Geom exhibits shifted peak positions and attenuated late-arriving components. These shifts indicate that boundary inaccuracies propagate into systematic path length errors in the downstream ray tracing. The delay domain shows a larger relative degradation than the angular domain. This asymmetry arises because a small boundary shift may preserve the approximate arrival angle of a reflected path while directly altering the total propagation distance. The resulting delay errors accumulate across multiple bounces. This pattern is consistent across all 100 test scenarios, as confirmed by the violin distributions in Fig.~\ref{fig:violin_aps} and Fig.~\ref{fig:violin_pdp}.

\subsection*{Pixel-level baselines fail at LoS boundary reconstruction}

RadioUNet and RMTransformer both formulate LoS prediction as dense pixel-level image-to-image translation. Their performance collapses across all channel dimensions. In RSS prediction, RadioUNet achieves a correlation of only 0.135 and RMTransformer achieves 0.171 (Table~\ref{tab:rss_comparison}). These values indicate that the predicted power maps are nearly uncorrelated with the ground truth. The violin distributions in Fig.~\ref{fig:violin_rss} further reveal heavy-tailed error distributions with large inter-scenario variance. Some scenarios produce biases exceeding $-50$~dB, corresponding to order-of-magnitude power prediction errors. The failure is equally severe in the angular and delay domains. Both baselines yield APS shape cosine values below 0.32 (Table~\ref{tab:aps_comparison}) and PDP shape cosine values below 0.18 (Table~\ref{tab:pdp_comparison}). These shape cosine values indicate that the predicted profiles are essentially random with respect to the ground truth.

\begin{figure}[ht]
  \centering
  \includegraphics[width=\linewidth, height=0.38\textheight, keepaspectratio]{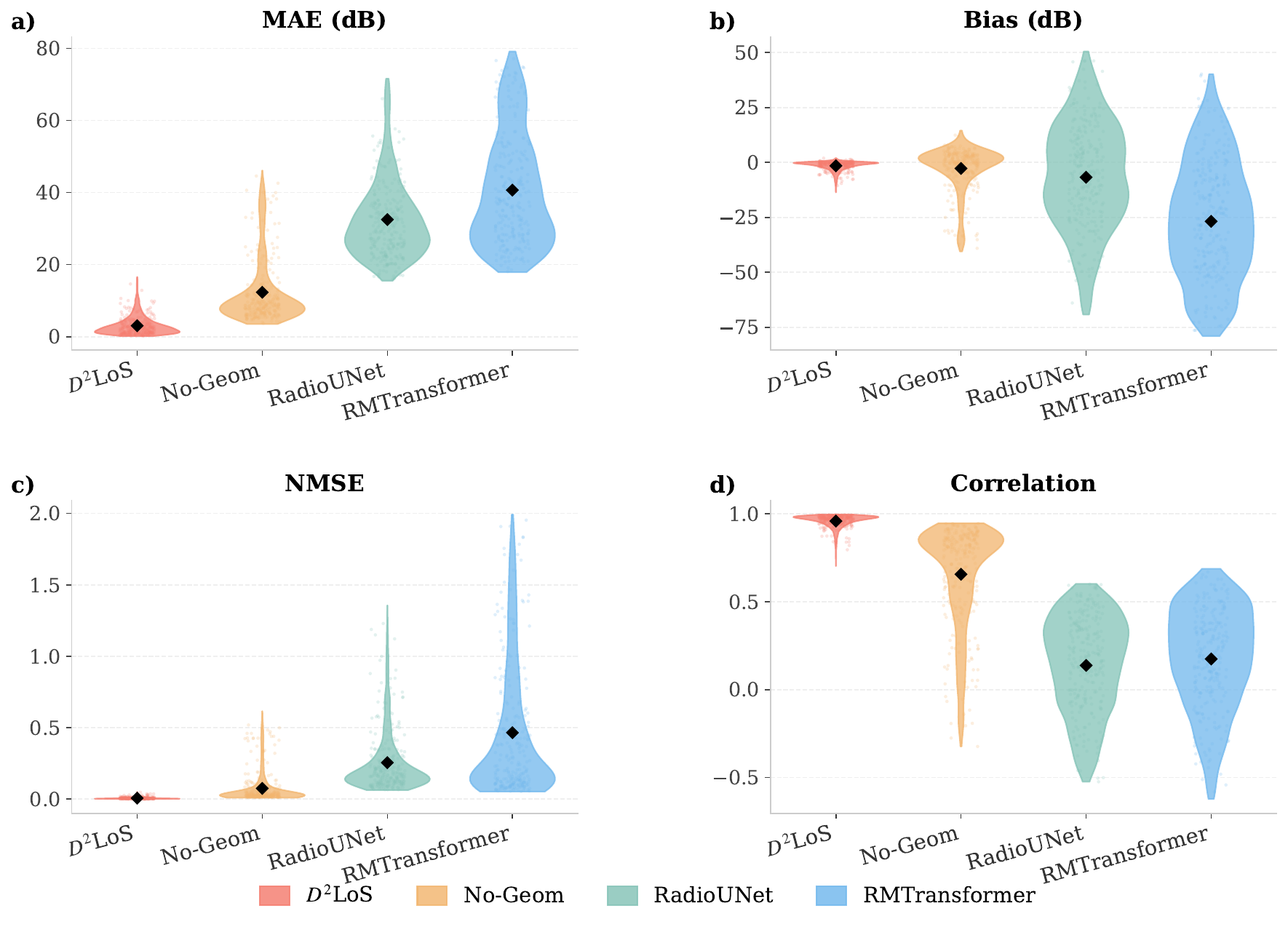}
  \caption{\textbf{Distribution of RSS prediction errors across 100 test scenarios.} Violin plots show the per-scenario distribution of six metrics: \textbf{a)} MAE, \textbf{b)} RMSE, \textbf{c)} bias, \textbf{d)} MSE, \textbf{e)} NMSE, and \textbf{f)} Pearson correlation. Black diamonds indicate the mean. D$^2$LoS exhibits both the lowest mean error and the smallest variance. RadioUNet and RMTransformer show wide distributions with heavy tails.}
  \label{fig:violin_rss}
\end{figure}

\begin{figure}[ht]
  \centering
  \includegraphics[width=\linewidth, height=0.38\textheight, keepaspectratio]{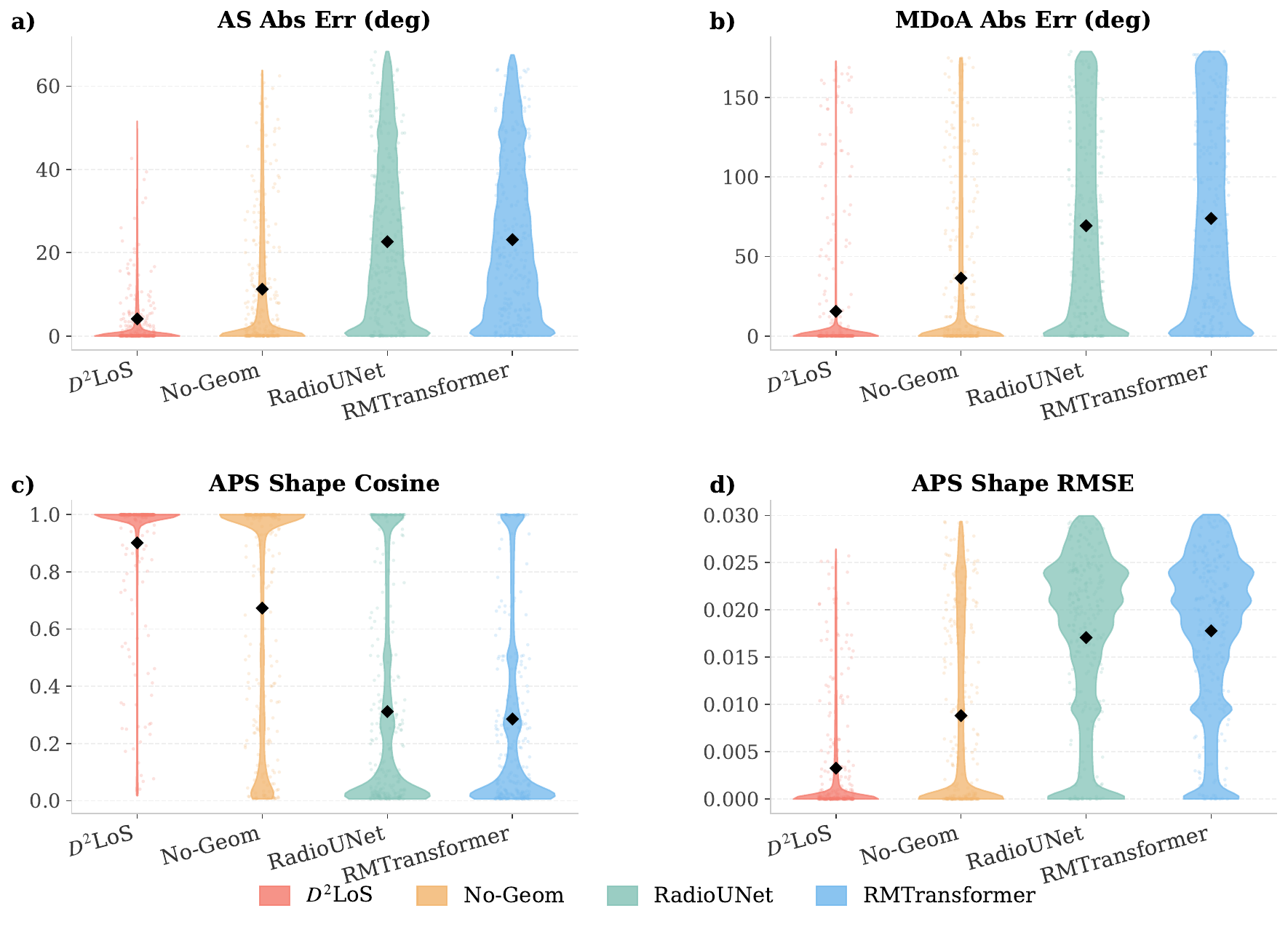}
  \caption{\textbf{Distribution of APS prediction errors across 100 test scenarios.} Violin plots show: \textbf{a)} angular spread error, \textbf{b)} MDoA error, \textbf{c)} APS shape cosine, and \textbf{d)} APS shape RMSE. D$^2$LoS achieves concentrated low-error distributions. The No-Geom variant shows moderate degradation. RadioUNet and RMTransformer produce near-random angular predictions.}
  \label{fig:violin_aps}
\end{figure}

\begin{figure}[ht]
  \centering
  \includegraphics[width=\linewidth, height=0.38\textheight, keepaspectratio]{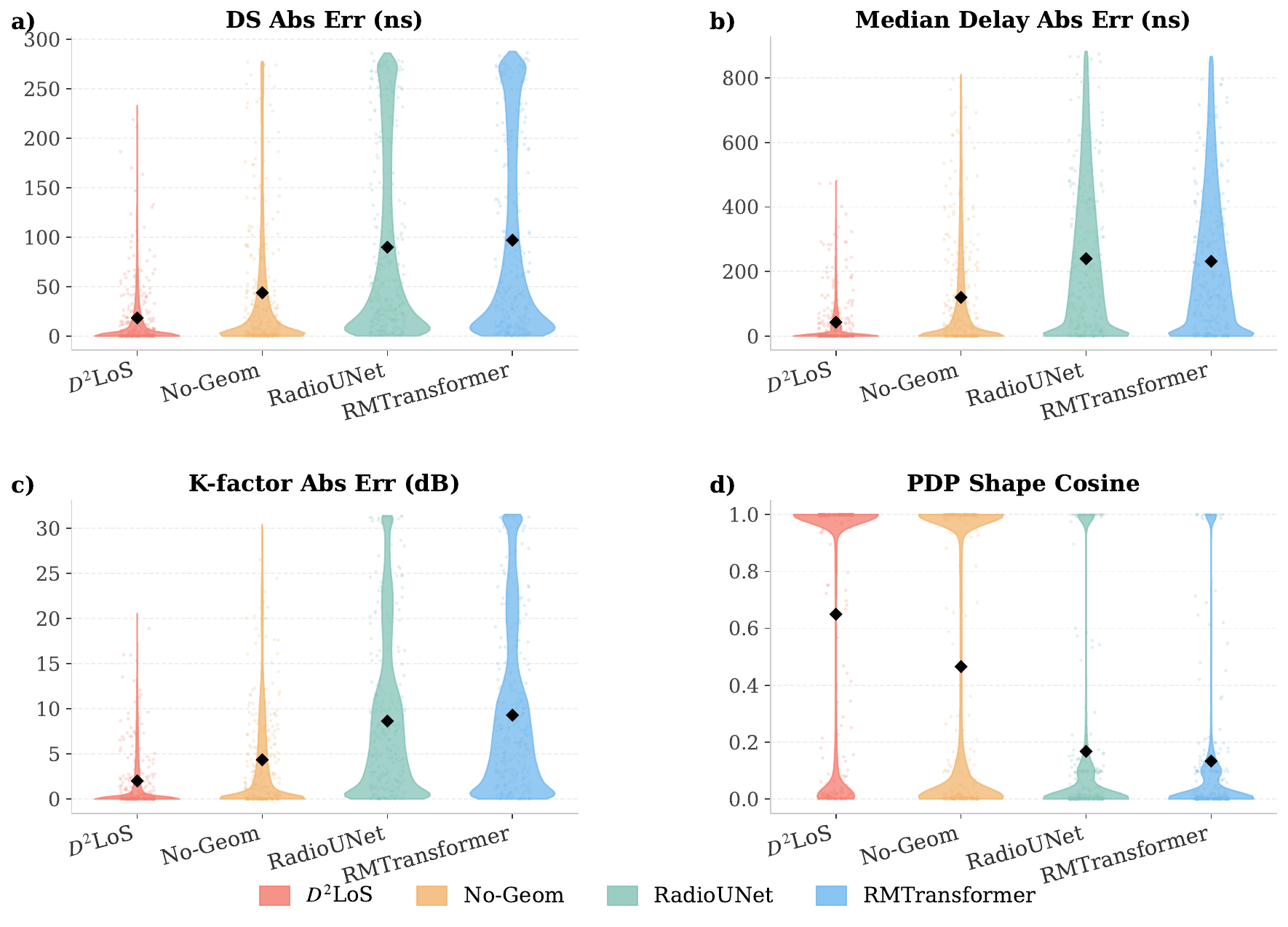}
  \caption{\textbf{Distribution of PDP prediction errors across 100 test scenarios.} Violin plots show: \textbf{a)} delay spread error, \textbf{b)} median delay error, \textbf{c)} $K$-factor error, \textbf{d)} effective count error, \textbf{e)} PDP shape cosine, and \textbf{f)} PDP shape RMSE. D$^2$LoS maintains compact error distributions. The gap between D$^2$LoS and No-Geom is larger in the delay domain than in the angular domain.}
  \label{fig:violin_pdp}
\end{figure}

The qualitative results reveal the underlying cause. The APS curves in Fig.~\ref{fig:aps} show noisy angular profiles with incorrect peak positions and power levels for both baselines. The PDP curves in Fig.~\ref{fig:pdp} exhibit spurious peaks and incorrect power floors, particularly at large delays where reflected and diffracted paths dominate. The RSS maps in Fig.~\ref{fig:rss} further illustrate this failure mode. RadioUNet produces heavily smoothed maps that erase all shadow boundary detail. RMTransformer generates narrow beam-like artifacts unrelated to physical propagation. These patterns are consistent with the spectral bias of neural networks, which favors low-frequency functions during training. The resulting blurred LoS boundaries remove the sharp spatial transitions that define shadow regions in urban environments. This spectral bias is the fundamental reason why pixel-level formulations cannot reconstruct the fine-grained LoS structure required for accurate channel modeling.

\subsection*{Per-ray angular information enables post-hoc beam pattern integration}

A key advantage of ray-level output is that arbitrary transmit antenna patterns can be applied after the propagation simulation without re-execution. For each retained ray, the framework stores the AoD in both azimuth and elevation. A directional antenna gain $G_\mathrm{ant}(\phi_\mathrm{AoD}, \theta_\mathrm{AoD})$ is then multiplied onto each ray's complex field before power summation. This produces beam-specific RSS, APS, and PDP maps from the same underlying ray data. We demonstrate this capability with a 15$^\circ$ half-power beamwidth pattern. Fig.~\ref{fig:rss_beam15} shows the resulting RSS maps for the same four scenarios as Fig.~\ref{fig:rss}. The directional pattern concentrates energy along the boresight and suppresses off-axis multipath components. D$^2$LoS accurately captures this directional filtering effect. It preserves both the beam footprint and the shadow structure behind buildings. No-Geom produces blurred beam edges due to inaccurate LoS boundaries. RadioUNet and RMTransformer fail to generate physically meaningful directional power distributions.

\begin{figure}[ht]
  \centering
  \includegraphics[width=\linewidth, height=0.38\textheight, keepaspectratio]{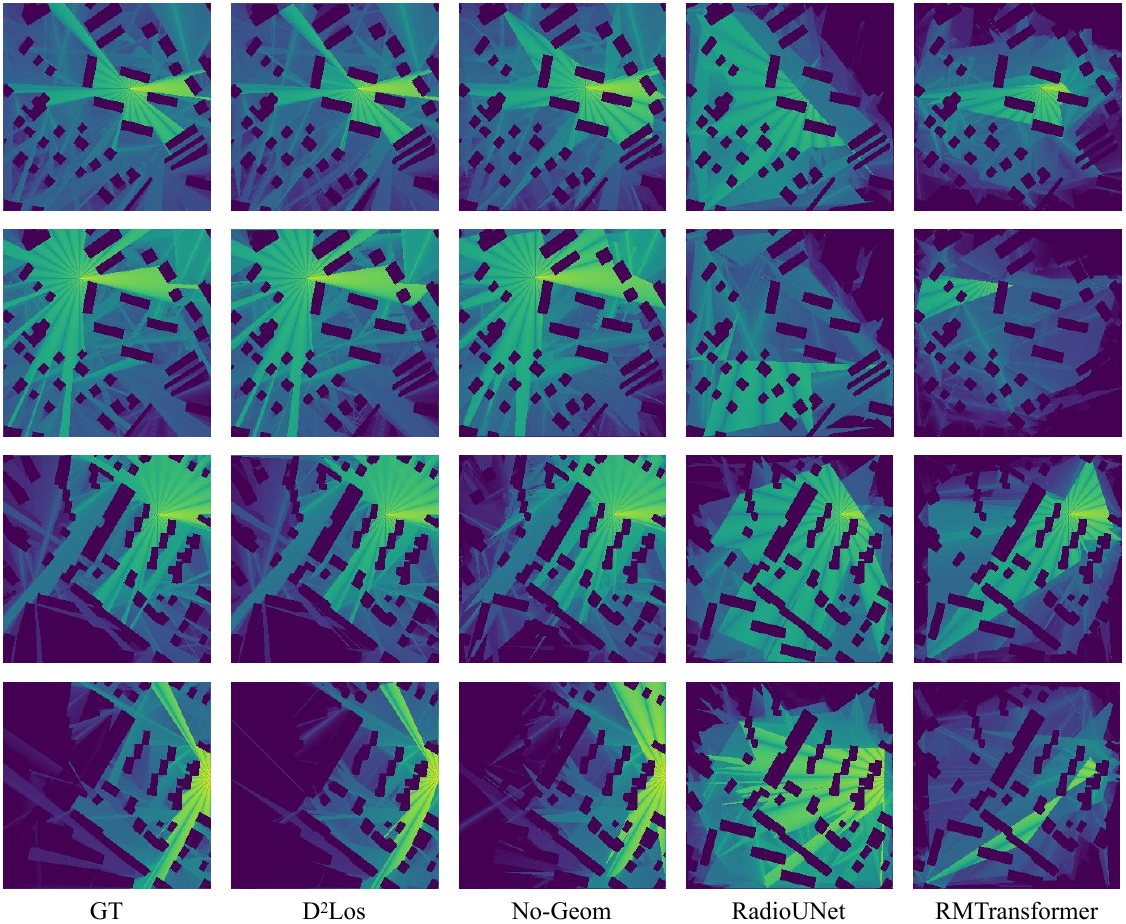}
  \caption{\textbf{RSS radio maps with a 15$^\circ$ directional beam pattern.} The same scenarios as Fig.~\ref{fig:rss} are shown after post-hoc beam integration. D$^2$LoS accurately captures the directional filtering effect, preserving both the beam footprint and shadow structure. RadioUNet and RMTransformer fail to produce meaningful directional patterns.}
  \label{fig:rss_beam15}
\end{figure}

\begin{figure*}[t]
  \centering
  \includegraphics[width=0.32\textwidth]{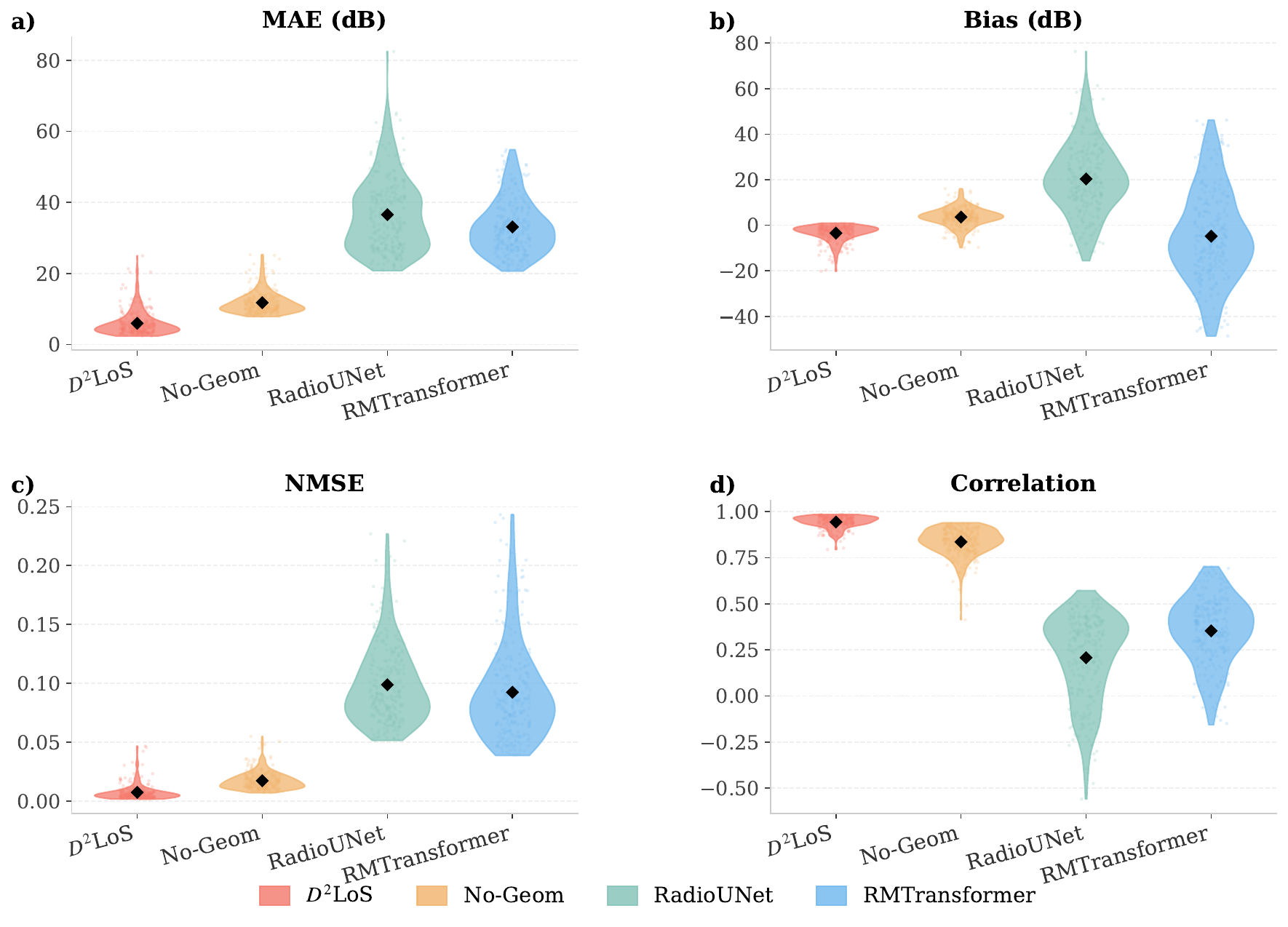}
  \includegraphics[width=0.32\textwidth]{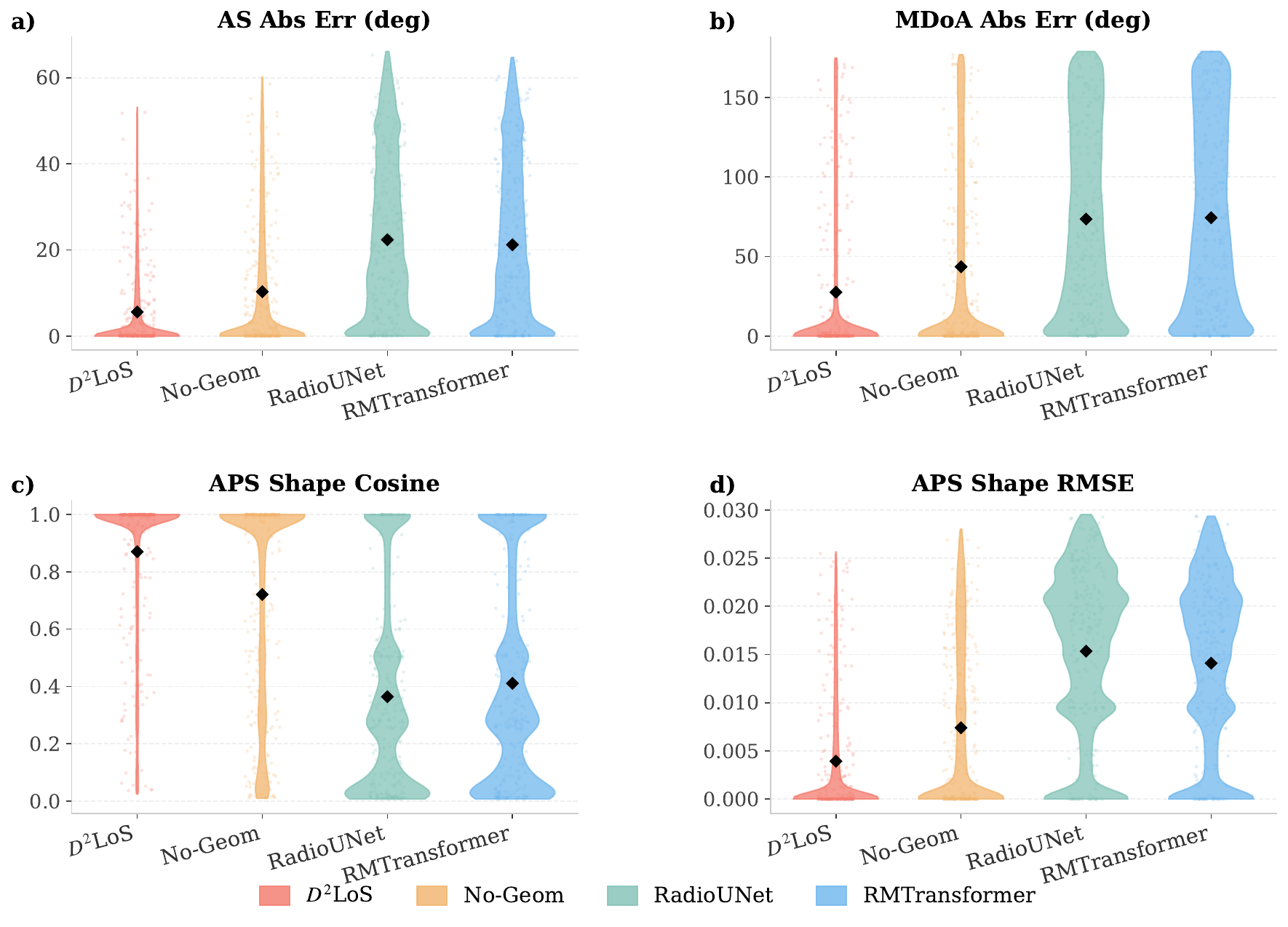}
  \includegraphics[width=0.32\textwidth]{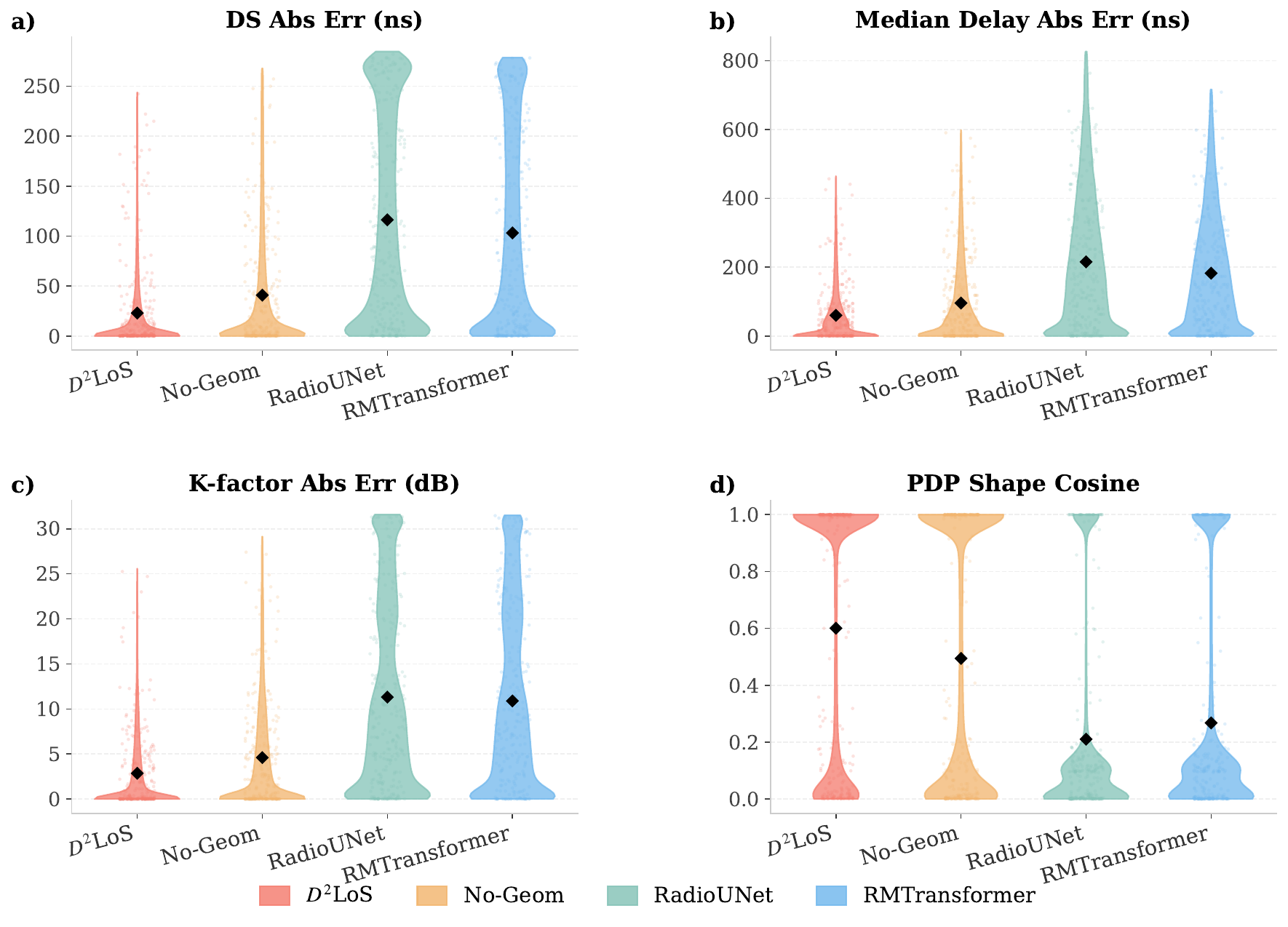}
  \caption{\textbf{Prediction accuracy with a 15$^\circ$ half-power beamwidth antenna pattern.} Violin plots show per-scenario error distributions after beam pattern integration using per-ray AoD. Left: RSS metrics. Center: APS metrics. Right: PDP metrics. D$^2$LoS maintains compact, low-error distributions consistent with the omnidirectional results. Additional beamwidths (30$^\circ$, 45$^\circ$, 60$^\circ$) are provided in the Supplementary Information.}
  \label{fig:violin_beam15}
\end{figure*}

The quantitative evaluation in Fig.~\ref{fig:violin_beam15} confirms these observations across all three channel dimensions. D$^2$LoS maintains compact, low-error distributions consistent with the omnidirectional results. The performance ranking among methods is preserved after beam integration. This indicates that applying a directional pattern does not amplify the relative errors introduced by LoS prediction inaccuracy. Additional results with 30$^\circ$, 45$^\circ$, and 60$^\circ$ beamwidths are provided in the Supplementary Information and show the same trend. This post-hoc integration capability is unavailable to conventional radio map methods that output only aggregate received power. Such methods do not preserve per-ray angular information and therefore cannot apply antenna patterns or evaluate beam-specific coverage after the fact. The ray-level output of D$^2$LoS thus offers a practical advantage for beam management and antenna deployment planning in 6G systems.

\subsection*{D$^2$LoS accelerates visibility computation}

We compare the end-to-end processing time of D$^2$LoS against the traditional RT pipeline. The traditional pipeline uses the rotational sweep algorithm for LoS preprocessing. Both pipelines share the same breadth-first search (BFS) path search \cite{beamer2013direction} and UTD field computation backend. All experiments are conducted on a single NVIDIA RTX Pro 6000 GPU. Fig.~\ref{fig:efficiency}a shows the per-scenario processing time for ten test maps. The traditional pipeline requires 15 to 48~min per scenario, while D$^2$LoS completes the same computation in under 1~min. The speedup factors range from 25$\times$ to 71$\times$, with the variation attributable to differences in scenario complexity. Fig.~\ref{fig:efficiency}b shows the cumulative distribution across all test scenarios. The median processing time is 0.5~min for D$^2$LoS and 24.9~min for the traditional pipeline, representing a median speedup of approximately 50$\times$.

\begin{figure}[ht]
  \centering
  \includegraphics[width=\linewidth, height=0.38\textheight, keepaspectratio]{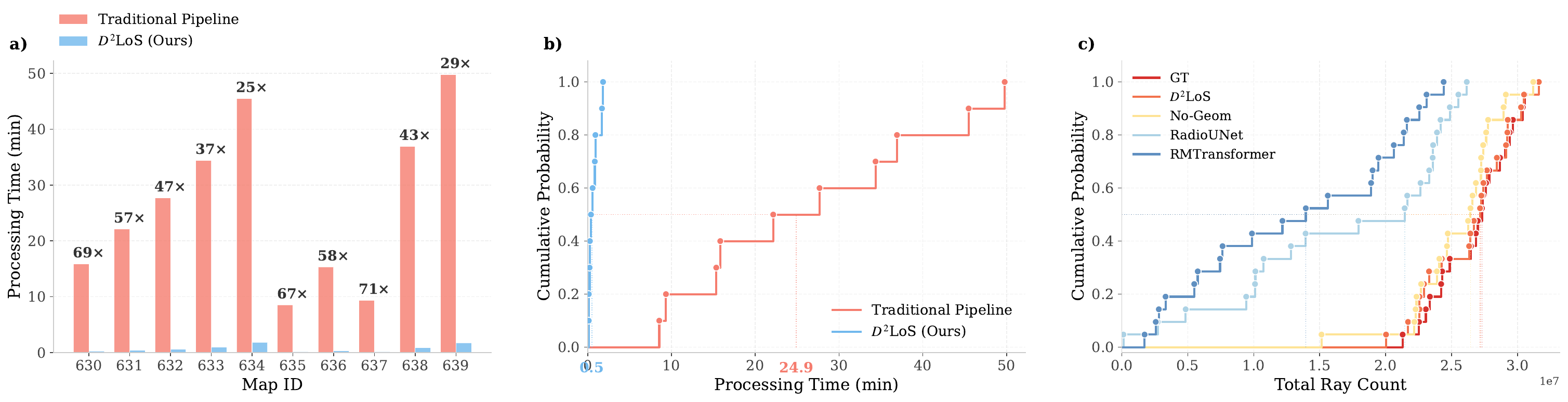}
  \caption{\textbf{Computational efficiency comparison between D$^2$LoS and the traditional RT pipeline.} \textbf{a)} Per-scenario processing time for ten test maps. Speedup factors range from 25$\times$ to 71$\times$. \textbf{b)} Cumulative distribution of processing time. The median is 0.5~min for D$^2$LoS and 24.9~min for the traditional pipeline. \textbf{c)} Cumulative distribution of the total ray count per scenario. All experiments use a single NVIDIA RTX Pro 6000 GPU.}
  \label{fig:efficiency}
\end{figure}

This acceleration is achieved without sacrificing output completeness. The D$^2$LoS pipeline produces the same ray-level output format as the traditional pipeline. Each ray stores complex gain, AoA, AoD, propagation delay, and complete geometric trajectory with 3D reflection-point coordinates. Fig.~\ref{fig:efficiency}c shows the cumulative distribution of the total ray count per scenario. The D$^2$LoS distribution closely matches the ground truth, confirming that the learned LoS prediction preserves the ray-level completeness of the full RT pipeline. Each ray is fully characterized for downstream applications, including coherent superposition, MIMO channel matrix synthesis, and cross-band CSI generation. The combination of over 25$\times$ speedup and full ray-level output makes D$^2$LoS a practical drop-in replacement for the LoS preprocessing stage in existing RT workflows.

\section{Discussion}

The results demonstrate that accurate LoS prediction is sufficient to produce high-fidelity multi-dimensional radio maps. In the ray tracing pipeline, LoS determination serves as the gating step for all downstream computations. A correct LoS map ensures that both the direct path and the first-level vertices of the BFS visibility tree are correctly identified. When LoS boundaries are inaccurate, errors propagate multiplicatively through the multi-bounce path search, corrupting all channel dimensions simultaneously. This mechanism explains why geometric post-processing has a stronger impact on delay-domain metrics than on angular-domain metrics: a small boundary shift may preserve the approximate arrival angle of a reflected path, but it directly alters the total propagation distance, producing systematic delay errors that accumulate across multiple bounces.

The complete failure of RadioUNet and RMTransformer highlights a fundamental limitation of pixel-level LoS prediction. Neural networks exhibit spectral bias, favoring low-frequency functions during training. The resulting blurred LoS boundaries remove the sharp shadow edges that define urban multipath structure. D$^2$LoS circumvents this by predicting sparse vertex attributes instead of dense pixel maps. Binary visibility labels and approximate projection coordinates are both low-frequency targets. The sharp LoS boundary is then reconstructed analytically, preserving the geometric precision that downstream ray tracing requires.

A key advantage of the LoS-based decomposition is its inherent frequency independence. The LoS computation depends only on building geometry, not on carrier frequency. Frequency enters solely through the UTD interaction coefficients, which can be updated without re-executing visibility computation or path search. This distinguishes D$^2$LoS from existing radio map methods that require retraining for each frequency band. Furthermore, the per-ray angular output enables post-hoc antenna pattern integration, as demonstrated by the beam-specific results across 15$^\circ$, 30$^\circ$, 45$^\circ$, and 60$^\circ$ beamwidths. The ability to evaluate arbitrary beam patterns from a single ray tracing computation provides substantial time savings for antenna deployment planning in multi-band 6G systems.

The RayVerse-100 dataset provides ray-level ground truth that existing datasets do not offer. Neither RadioMapSeer nor DeepMIMO exposes explicit per-ray geometric trajectories with 3D reflection-point coordinates. This trajectory information is essential for environment-aware beam tracking, RIS placement optimization, and radar-communication co-design.

Several limitations should be noted. The 2.5D building model introduces errors for complex rooftop geometries. The 2D$\times$2D decomposition excludes over-rooftop diffraction paths. The evaluation is conducted at street-level height (1.5~m) and retains only the top-8 rays per point. The geometric post-processing uses a fixed search radius that may benefit from adaptive scaling. Future work will extend the building model, incorporate over-rooftop diffraction, and integrate measured channel data through transfer learning to bridge the gap between simulation and real-world deployment.

\section{Methods}
 
\subsection*{Environment representation and problem formulation}
\begin{figure*}[t]
    \centering
    \includegraphics[width=\textwidth]{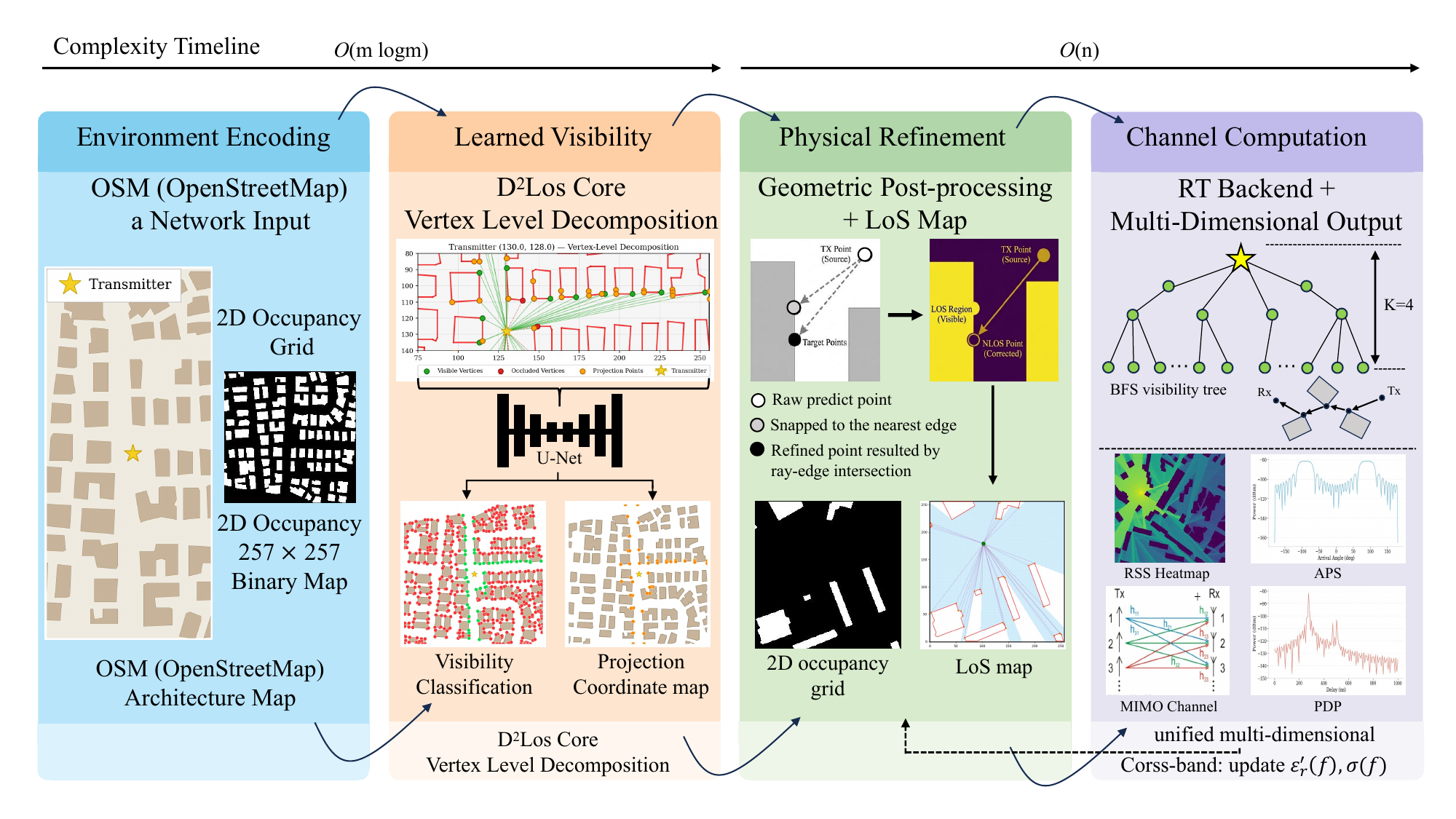}
    \caption{\textbf{Overview of the D$^2$LoS framework.} The pipeline takes a binary occupancy map and transmitter position as input. The D$^2$LoS network predicts vertex-level visibility and projection coordinates. Geometric post-processing enforces physical constraints to produce an exact LoS map. A BFS path search and UTD field computation yield per-ray channel parameters. Cross-band information is obtained by updating material parameters without re-executing visibility computation.}
    \label{fig:framework}
\end{figure*}
The propagation environment is encoded as a top-down binary occupancy map $E \in \{0,1\}^{H \times W}$ with $H = W = 257$ pixels at a spatial resolution of $\Delta r = 1$~m. Building footprints are extracted from OpenStreetMap (OSM) and approximated as quadrilateral prisms. The vertices of each footprint polygon serve as candidate reflection and diffraction points. Per-building heights are obtained from OSM attributes when available; a default value of 20~m is assigned otherwise.
 
Given a transmitter at $\mathbf{p}_s = (x_s, y_s)$, the goal is to determine a binary LoS map $L \in \{0,1\}^{H \times W}$, where $L(x_r, y_r) = 1$ indicates that pixel $(x_r, y_r)$ is visible from $\mathbf{p}_s$. The network input is a $257 \times 257 \times 4$ tensor comprising: the binary occupancy map $E$, a Gaussian-blurred transmitter location heatmap $H_\mathrm{Tx}$ with $\sigma = 3$ pixels, and two spatial coordinate grids $X_\mathrm{coord}$ and $Y_\mathrm{coord}$.
 
The ray tracing pipeline is decomposed into two stages (Fig.~\ref{fig:framework}): LoS preprocessing, which determines mutual visibility among all geometric primitives, and BFS multi-bounce path search, which enumerates valid propagation paths. D$^2$LoS replaces the LoS preprocessing stage while preserving the standard path search and UTD field computation backend.

\subsection*{Vertex-level LoS reformulation}
 
In a 2D environment composed of polygonal obstacles, the LoS region boundary as seen from a point source $\mathbf{p}_s$ is piecewise-linear. Each linear segment is defined by a building vertex $\mathbf{v}_i$ visible from $\mathbf{p}_s$ and its projection point $\mathbf{v}_i^\mathrm{proj}$ on a more distant building edge. The projection point is computed via ray casting: a ray from $\mathbf{p}_s$ passes through $\mathbf{v}_i$ and extends until it first intersects a building edge $e_j$ that is farther from $\mathbf{p}_s$ and not parallel to the ray:
\begin{equation}
\mathbf{v}_i^\mathrm{proj} = \mathbf{p}_s + t_i^* \cdot (\mathbf{v}_i - \mathbf{p}_s), \quad
t_i^* = \min_j \big\{ t > 1 \;\big|\; (\mathbf{p}_s + t \cdot \mathbf{d}_i) \in e_j,\; \mathbf{d}_i \times \mathbf{n}_{e_j} \neq 0 \big\},
\label{eq:projection}
\end{equation}
where $\mathbf{d}_i = \mathbf{v}_i - \mathbf{p}_s$ and $\mathbf{n}_{e_j}$ is the direction vector of edge $e_j$.
 
Two prediction targets are defined for each building vertex $\mathbf{v}_i$: a binary visibility label $c_i \in \{0, 1\}$ indicating whether $\mathbf{v}_i$ is visible from $\mathbf{p}_s$, and the projection point coordinates $\hat{\mathbf{v}}_i^\mathrm{proj} = (\hat{x}_i^\mathrm{proj}, \hat{y}_i^\mathrm{proj})$, defaulting to $\mathbf{v}_i$ for non-boundary vertices. All $M$ building vertices are predicted simultaneously, avoiding the need to pre-select which vertices determine the LoS boundary. The network predicts low-frequency attributes rather than a dense high-frequency boundary map. The sharp LoS boundary is reconstructed analytically from the predicted vertex attributes in the post-processing stage.

\subsection*{Network architecture and loss function}
 
We employ a U-Net encoder-decoder with two modifications at the bottleneck layer. Atrous spatial pyramid pooling (ASPP) applies dilated convolutions with rates $\{1, 6, 12, 18\}$ to capture occlusion effects at multiple spatial scales. Spectral feature modulation follows the Fourier neural operator (FNO) paradigm: spatial features are transformed to the frequency domain via real FFT, modulated by learned complex-valued weights, and transformed back via inverse FFT. This provides a global receptive field in a single layer and captures long-range shadow connectivity that local convolutions cannot efficiently model.
 
The decoder produces two output heads via sub-pixel upsampling with coordinate attention: a visibility classification map $\hat{M} \in [0,1]^{H \times W}$ and normalised projection coordinates $\hat{P} \in [0,1]^{H \times W \times 2}$. During inference, predictions are extracted at known building vertex pixel locations, effectively implementing sparse vertex-level prediction through a dense network architecture.
 
The composite loss function is
\begin{equation}
\mathcal{L}_\mathrm{total} = \mathcal{L}_\mathrm{Focal} + \lambda_1 \mathcal{L}_\mathrm{Dice} + \lambda_2 \mathcal{L}_\mathrm{MaskedL1}.
\label{eq:loss}
\end{equation}
Focal loss with $\gamma = 2.0$ addresses class imbalance between visible and occluded vertices. Dice loss with $\lambda_1 = 1.0$ enforces spatial connectivity of predicted LoS regions. The coordinate regression head uses a masked L1 loss computed only at building-vertex pixels:
\begin{equation}
\mathcal{L}_\mathrm{MaskedL1} = \frac{\sum \|\hat{P} - P^\mathrm{GT}\|_1 \cdot \mathbf{1}(M^\mathrm{GT} > 0.5)}{\sum \mathbf{1}(M^\mathrm{GT} > 0.5)},
\label{eq:maskedl1}
\end{equation}
ensuring that representational capacity concentrates on geometrically meaningful boundary regions. We set $\lambda_2 = 20.0$ based on hyperparameter search.

\subsection*{Geometric post-processing and 3D extension}
 
Raw neural network predictions contain residual regression errors in the projection coordinates. We exploit the physical constraint that projection points must lie on building edges through a two-step refinement.
 
In Step~1, for each predicted projection point $\hat{\mathbf{v}}_i^\mathrm{proj}$, the nearest building edge $e_j^*$ within a search radius $R_\mathrm{search} = 5$~m is identified. In Step~2, the intersection of the ray from $\mathbf{p}_s$ through vertex $\mathbf{v}_i$ with edge $e_j^*$ is computed:
\begin{equation}
\mathbf{v}_i^\mathrm{refined} = \mathbf{p}_s + t^* (\mathbf{v}_i - \mathbf{p}_s), \quad
t^* = \frac{(\mathbf{a} - \mathbf{p}_s) \times (\mathbf{b} - \mathbf{a})}{(\mathbf{v}_i - \mathbf{p}_s) \times (\mathbf{b} - \mathbf{a})},
\label{eq:refinement}
\end{equation}
where $\mathbf{a}$ and $\mathbf{b}$ are the endpoints of $e_j^*$. The refined point is accepted if $t^* > 1$ and the intersection lies within the edge segment; otherwise, the raw prediction is retained. A global path-clearance check verifies that the refined ray does not intersect other buildings.
 
The complete LoS map is reconstructed by selecting boundary vertices $V_\mathrm{bnd} = \{\mathbf{v}_i \mid c_i = 1,\, \mathbf{v}_i^\mathrm{refined} \neq \mathbf{v}_i\}$, constructing segments $[\mathbf{v}_i, \mathbf{v}_i^\mathrm{refined}]$, sorting them angularly with respect to $\mathbf{p}_s$, taking the union of unobstructed angular sectors, and rasterising onto the grid to produce $\hat{L} \in \{0,1\}^{H \times W}$.
 
To extend to 3D visibility, we adopt the 2D$\times$2D decomposition:
\begin{equation}
\mathrm{LoS}_\mathrm{3D}(\mathbf{p}_s, \mathbf{p}_r) \approx \mathrm{LoS}_\mathrm{2D\text{-}H}(\mathbf{p}_s, \mathbf{p}_r) \;\wedge\; \mathrm{LoS}_\mathrm{2D\text{-}V}(\mathbf{p}_s, \mathbf{p}_r).
\label{eq:3d}
\end{equation}
The horizontal check $\mathrm{LoS}_\mathrm{2D\text{-}H}$ is computed by D$^2$LoS. The vertical check $\mathrm{LoS}_\mathrm{2D\text{-}V}$ evaluates whether the elevation angle of the direct path clears the heights of all intervening buildings. Under the 2.5D extruded prism model, this decomposition is mathematically exact (Supplementary Theorem~S3). The transmitter and receiver heights are input parameters; the current evaluation uses $h_\mathrm{tx} = h_\mathrm{rx} = 1.5$~m. Over-rooftop diffraction paths are excluded from the current framework, as discussed in the limitations.

\subsection*{Ray tracing backend}
 
With LoS relationships determined, an $N$-ary visibility tree is constructed rooted at the transmitter. Vertex-to-vertex LoS adjacency is precomputed once per scenario using the exact rotational sweep algorithm. D$^2$LoS provides transmitter-to-vertex visibility, forming the first level of the tree. Multi-bounce paths are enumerated via breadth-first search (BFS) with a maximum depth of $K = 4$, covering up to four combined reflections and diffractions. The search is executed with GPU-parallel graph traversal.
 
For each valid path, the complex electric field is computed as the product of interaction coefficients and free-space propagation factors along all segments. Interaction coefficients include Fresnel reflection coefficients and UTD diffraction coefficients following the Luebbers lossy-dielectric extension. Two numerical stability mechanisms are applied: an envelope clamp ensuring that diffracted field magnitude does not exceed half the incident field, and a forward-scatter smoothing for deflection angles below 30$^\circ$ (Supplementary Theorem~S2 and Supplementary Remark~1). Default material parameters at the reference frequency of 3.5~GHz are $\epsilon_r' = 5.31$ and $\sigma = 0.0326$~S/m.
 
At each evaluation point, all arriving rays are ranked by path gain and the top $N_\mathrm{ray} = 8$ rays are retained, capturing over 98\% of the total received energy across all scenarios. Each retained ray stores: complex gain, AoA and AoD in azimuth and elevation, propagation delay, complete geometric trajectory with 3D interaction-point coordinates, interaction type at each point, and material index of each interacting surface. The RSS output is computed via incoherent power summation over all retained rays. Because per-ray phase information is fully preserved, coherent superposition is available through the post-processing interface.
 
The geometric ray paths are frequency-independent. To obtain channel information at a different carrier frequency, users specify the target frequency and corresponding material parameters $\epsilon_r'(f)$ and $\sigma(f)$. The framework re-evaluates all interaction coefficients along the stored geometric trajectories without re-executing visibility computation or path search. Similar to \cite{alkhateeb2019deepmimo}, the narrowband MIMO channel matrix for an $N_t \times N_r$ array system is
\begin{equation}
\mathbf{H}(f) = \sum_{p=1}^{N_\mathrm{ray}} \alpha_p \, \mathbf{a}_R(\theta_{R,p}, \phi_{R,p}) \, \mathbf{a}_T^H(\theta_{T,p}, \phi_{T,p}) \, e^{-j 2\pi f \tau_p},
\label{eq:mimo}
\end{equation}
where $\alpha_p$ is the complex path gain, $\tau_p$ is the delay, and $\mathbf{a}_R$, $\mathbf{a}_T$ are the receive and transmit array steering vectors.

\subsection*{Dataset construction and training configuration}
 
To train D$^2$LoS, we extracted 701 urban topographies from OSM covering diverse building densities and street layouts. Each environment is rasterised into a $257 \times 257$ grid. For each scenario, multiple transmitter positions are sampled and the ground-truth LoS map and projection point coordinates are computed using the exact rotational sweep algorithm. The 701 scenarios are strictly partitioned: scenarios 0--630 for training and validation, and scenarios 631--700 as a completely unseen test set.
 
To provide ray-level ground truth for end-to-end evaluation, we constructed RayVerse-100 by selecting 100 scenarios from geographically distinct regions. For each scenario, the full RT pipeline is executed: exact LoS computation, BFS path enumeration with $K = 4$, and UTD field computation at every grid point. The output contains the top-8 per-ray multipath parameters. The RayVerse-100 scenarios are drawn from the test partition to ensure no overlap with D$^2$LoS training data. Built on UTD, the dataset exposes open material parameter interfaces for cross-band CSI generation. Post-processing interfaces for coherent superposition, MIMO channel matrix synthesis, and antenna pattern integration are provided.
 
D$^2$LoS is trained on NVIDIA RTX 4090 and A100 GPUs using AdamW with a cosine annealing learning rate schedule. The initial learning rate is $1.8 \times 10^{-3}$, the batch size is 64, and training runs for 120 epochs. All metrics reported in Tables~\ref{tab:rss_comparison}--\ref{tab:pdp_comparison} are computed as mean $\pm$ standard deviation over the 100 independent test scenarios. No scenario in the test set overlaps with the training data geographically or topologically.

\begin{appendices}
\section*{Supplementary Theorem S1: Complexity Analysis of D$^2$LoS}

\begin{theorem}[Computational complexity reduction]
\label{thm:complexity}
Given $m$ building vertices and $n$ evaluation points, D$^2$LoS reduces the per-transmitter LoS preprocessing complexity from $O(m \log m + n)$ to $O(M \log M + n)$, where $M$ is the number of predicted boundary vertices. In typical urban scenarios, $M \ll m$ and $M \log M \ll n$, so the effective per-transmitter cost is $O(n)$. Under GPU parallelisation with sufficient cores, the wall-clock time is bounded by a constant independent of both $m$ and $n$.
\end{theorem}

\begin{proof}
\textit{Classical baseline.}
The rotational sweep algorithm computes the complete LoS polygon for a single transmitter in $O(m \log m)$: all $m$ edges are sorted angularly and swept. Evaluation of $n$ receiver points against this polygon costs $O(n)$ via angular binary search. The total per-transmitter cost is $O(m \log m + n)$.

\textit{D$^2$LoS cost decomposition.}
D$^2$LoS replaces the $O(m \log m)$ polygon construction with three operations.

\textit{Step 1: CNN inference.}
The network processes a fixed $H \times W$ input tensor with $H = W = 257$. For a U-Net with $L$ layers each of bounded filter size, the computation cost is $O(L \cdot H \cdot W)$. Since $L$, $H$, and $W$ are all fixed constants independent of $m$ and $n$, the total inference cost is a constant $C_\mathrm{CNN}$.

\textit{Step 2: Geometric post-processing.}
For each of the $M$ predicted boundary vertices, we identify the nearest building edge within a search radius $R_\mathrm{search}$. We assume a uniform grid-based spatial index is precomputed for each scenario in $O(|E|)$ time, where $|E|$ is the total number of building edges. Given this index, each query retrieves at most $k$ candidate edges, where $k \leq \pi R_\mathrm{search}^2 \cdot \rho_\mathrm{edge}$ and $\rho_\mathrm{edge}$ is the maximum local edge density. Since $\rho_\mathrm{edge}$ is bounded by the building geometry and $R_\mathrm{search}$ is a fixed constant, $k$ is bounded by a constant independent of $m$. Each ray-edge intersection test costs $O(1)$. The total cost of geometric post-processing is therefore $O(M \cdot k) = O(M)$.

\textit{Step 3: LoS boundary rasterisation.}
The $M$ boundary segments are sorted angularly with respect to $\mathbf{p}_s$ in $O(M \log M)$. The sorted boundary is then rasterised onto the $H \times W$ grid by scanline traversal in $O(n)$, where $n = H \times W$. The total rasterisation cost is $O(M \log M + n)$.

\textit{Total cost.}
Combining the three steps:
\begin{equation}
T_{\mathrm{D}^2\mathrm{LoS}} = C_\mathrm{CNN} + O(M) + O(M \log M + n) = O(M \log M + n).
\tag{S1}
\end{equation}

In typical urban scenarios, $m \approx 1500$--$3000$ and $M \approx 20$--$80$ (i.e., only 2--5\% of building vertices lie on the LoS boundary for a given transmitter). Therefore $M \log M \leq 80 \times \log_2(80) \approx 506$, while $n = 257^2 = 66{,}049$. Since $M \log M \ll n$, the effective complexity simplifies to $O(n)$, compared with $O(m \log m + n)$ for the classical method. The improvement factor in the polygon construction step is $m \log m / (M \log M) \geq 30\times$ for typical values.

\textit{GPU parallelisation.}
The three steps parallelise as follows.
CNN inference maps to tensor-core parallelism with wall-clock time determined by network depth, not spatial resolution. This is a constant $T_\mathrm{CNN}$ on any modern GPU.
The $M$ edge-snapping and ray-edge intersection operations are independent and execute in $O(1)$ wall-clock time with $M$ parallel threads.
Angular sorting of $M$ elements can be performed via parallel merge sort in $O(\log^2 M)$ wall-clock time.
The $n$ point-in-polygon evaluations are independent; each tests membership against the sorted angular boundary in $O(\log M)$ via binary search. With $n$ parallel threads, this step has $O(\log M)$ wall-clock time.

The total wall-clock time under GPU parallelisation is:
\begin{equation}
T_\mathrm{wall} = T_\mathrm{CNN} + O(1) + O(\log^2 M) + O(\log M) = O(T_\mathrm{CNN} + \log^2 M).
\tag{S2}
\end{equation}

Since both $T_\mathrm{CNN}$ and $\log^2 M$ are bounded constants independent of $m$ and $n$, the wall-clock time per transmitter is $O(1)$ with respect to the scenario size.
\end{proof}

\section*{Supplementary Theorem S2: UTD Numerical Stability and Field Continuity}

\begin{theorem}[Pointwise diffracted field bound and shadow boundary continuity]
\label{thm:utd}
Consider the Luebbers UTD diffraction coefficient $D$ for a lossy dielectric wedge with exterior angle $n\pi$. The following properties hold:

(a) At the shadow boundary ($\phi - \phi' = \pi$), the diffracted field satisfies $|E_d| = 0.5 \cdot |E_i|$, ensuring that the total field $E_\mathrm{total} = E_i + E_d$ is continuous across the boundary.

(b) For all observation angles, $|D|$ is bounded above by its value at the shadow boundary. Therefore, the envelope clamp $|D|_\mathrm{max} = 0.5$ does not alter the physical solution and serves only as a numerical safeguard.
\end{theorem}

\begin{proof}
\textit{Part (a): Shadow boundary limit.}

The Luebbers UTD diffraction coefficient consists of four terms:
\begin{equation}
D = \sum_{i=1}^{4} D_i, \quad
D_i = \frac{-e^{-j\pi/4}}{2n\sqrt{2\pi k}} \, C_i \cdot \cot\!\left(\frac{\pi + \beta_i}{2n}\right) F\!\left(kL_i \, a^+(\beta_i)\right),
\tag{S3}
\end{equation}
where $\beta_i \in \{\phi - \phi',\; \phi + \phi'\}$ with appropriate sign combinations, $C_i \in \{1, 1, R_0, R_n\}$ are the face reflection coefficients, and $F(X) = 2j\sqrt{X}\, e^{jX} \int_{\sqrt{X}}^{\infty} e^{-jt^2}\, dt$ is the Fresnel transition function.

At the incident shadow boundary where $\beta_1 = \phi - \phi' \to \pi$, the cotangent factor in $D_1$ has a pole: $\cot\!\big(\frac{\pi + \beta_1}{2n}\big) \to \infty$. Simultaneously, the transition function argument satisfies $a^+(\beta_1) \to 0$, so $F(kL_1 \, a^+(\beta_1)) \to 0$. Applying L'H\^opital's rule to the product:
\begin{equation}
\lim_{\beta_1 \to \pi} \cot\!\left(\frac{\pi + \beta_1}{2n}\right) F\!\left(kL_1 \, a^+(\beta_1)\right) = n \sqrt{\frac{2\pi k}{L_1}}\, e^{j\pi/4}.
\tag{S4}
\end{equation}

Substituting back into $D_1$ and multiplying by the free-space spreading factor $\sqrt{L_1}/(4\pi d)$ yields $|E_d^{(1)}| \to 0.5 \cdot |E_i|$ at the shadow boundary. The remaining terms $D_2$, $D_3$, $D_4$ involve $\phi + \phi'$ combinations that do not produce poles at $\phi - \phi' = \pi$ and contribute finite corrections. Their sum preserves the total limit of $0.5 \cdot |E_i|$ by the standard UTD continuity condition, which requires the total field to equal $0.5 \cdot E_i$ at the shadow boundary by construction of the Fresnel transition function.

\textit{Part (b): Global bound.}

Away from the shadow boundary, $a^+(\beta_i) > 0$ and the Fresnel transition function satisfies $|F(X)| \leq 1$ for all $X \geq 0$, with $|F(X)| \to 1$ as $X \to \infty$. Meanwhile, the cotangent factors remain bounded when $\beta_i$ is away from $\pm n\pi$. The product $|\cot(\cdot) \cdot F(\cdot)|$ achieves its maximum at the shadow boundary where the L'H\^opital limit applies, and decreases monotonically as the observation angle moves away from the boundary. This monotonic decay is a well-known property of the UTD diffraction coefficient for convex wedges with $n \geq 1$ (see Kouyoumjian and Pathak, Proc.\ IEEE, 1974).

Therefore, $|D(\phi)| \leq |D(\phi = \phi' + \pi)|$ for all observation angles $\phi$. The envelope clamp $|D|_\mathrm{max} = 0.5$ does not clip any physically correct solution. It acts solely as a guard against floating-point overflow when $\cot(\cdot)$ and $F(\cdot)$ are evaluated independently near the shadow boundary, where their product is finite but each factor individually diverges.

\textit{Numerical implementation.}
In floating-point arithmetic, finite-precision cancellation in the transition function $F(X)$ for small $X$ can produce spurious large values before the product is formed. The envelope clamp provides a post-hoc correction:
\begin{equation}
|E_d^\mathrm{clamped}| = \min\!\big(|E_d^\mathrm{computed}|,\; 0.5 \cdot |E_i|\big).
\tag{S5}
\end{equation}
This ensures that numerical artifacts do not violate the physical bound established in Part~(b).
\end{proof}

\begin{remark}[Forward-scatter smoothing]
\label{rem:forward}
For near-grazing forward scattering with deflection angle $0 < \theta < 30^\circ$, the four UTD cotangent terms partially cancel, producing a small net diffraction coefficient with rapid oscillation due to the competing signs of $D_1$--$D_4$. These oscillations are sensitive to small perturbations in wedge geometry and are not reliably reproduced in discretised building models where edge positions are known only to pixel-level precision ($\Delta r = 1$~m). We replace the oscillatory solution with a monotonic linear attenuation in the log domain:
\begin{equation}
L_\mathrm{smooth}(\theta) = 6.02 + \frac{30 - 6.02}{30^\circ} \cdot (30^\circ - \theta) \quad [\mathrm{dB}], \quad \theta \in [0^\circ, 30^\circ].
\tag{S6}
\end{equation}
At $\theta = 30^\circ$, $L_\mathrm{smooth} = 6.02$~dB, matching the UTD envelope clamp and ensuring $C^0$ continuity. At $\theta = 0^\circ$, $L_\mathrm{smooth} = 30$~dB, reflecting the physical expectation of strong attenuation for forward-diffracted rays. This smoothing is a modelling choice validated by the end-to-end RSS accuracy reported in the main text.
\end{remark}

\section*{Supplementary Theorem S3: Geometric Validity of 2D$\times$2D Decomposition}

\begin{theorem}[Exactness of 2D$\times$2D LoS decomposition]
\label{thm:2d2d}
Consider a set of buildings $\{B_1, \ldots, B_N\}$ under the 2.5D model, where each building is an extruded right prism $B_i = P_i \times [0, h_i]$ with non-overlapping footprints ($P_i \cap P_j = \emptyset$ for $i \neq j$), polygonal footprint $P_i$, and height $h_i > 0$. For any transmitter $\mathbf{p}_s = (x_s, y_s, h_\mathrm{tx})$ and receiver $\mathbf{p}_r = (x_r, y_r, h_\mathrm{rx})$ located outside all buildings, the full 3D direct-path LoS test is equivalent to:
\begin{equation}
\mathrm{LoS}_\mathrm{3D}(\mathbf{p}_s, \mathbf{p}_r) = \mathrm{LoS}_\mathrm{2D\text{-}H}(\mathbf{p}_s, \mathbf{p}_r) \;\wedge\; \mathrm{LoS}_\mathrm{2D\text{-}V}(\mathbf{p}_s, \mathbf{p}_r),
\tag{S7}
\end{equation}
where $\mathrm{LoS}_\mathrm{2D\text{-}H}$ is the horizontal intersection test on the 2D projected ray and $\mathrm{LoS}_\mathrm{2D\text{-}V}$ is the vertical elevation test against building heights. This equivalence holds with zero approximation error.
\end{theorem}

\begin{proof}
\textit{Setup.}
Let $R_\mathrm{3D}(t) = \big(r_\mathrm{2D}(t),\, z(t)\big)$ parameterise the direct ray from $\mathbf{p}_s$ to $\mathbf{p}_r$ with $t \in [0,1]$, where $r_\mathrm{2D}(t) = (1-t)(x_s, y_s) + t(x_r, y_r)$ and $z(t) = (1-t)\, h_\mathrm{tx} + t\, h_\mathrm{rx}$. Note that $z(t)$ is linear and therefore monotonic or constant over any sub-interval.

\textit{Necessary and sufficient condition for occlusion.}
The ray is occluded by building $B_i$ if and only if there exists $t^* \in (0,1)$ such that $R_\mathrm{3D}(t^*) \in B_i$. By the Cartesian product structure $B_i = P_i \times [0, h_i]$, this decomposes as:
\begin{equation}
\exists\, t^* \in (0,1) : \quad r_\mathrm{2D}(t^*) \in P_i \;\wedge\; 0 \leq z(t^*) \leq h_i.
\tag{S8}
\end{equation}

\textit{Case 1: No horizontal intersection.}
If $r_\mathrm{2D}(t) \notin P_i$ for all $t \in (0,1)$, then no $t^*$ satisfying Eq.~(S8) exists, and the ray clears $B_i$. No vertical check is needed.

\textit{Case 2: Horizontal intersection exists.}
Suppose the 2D projected ray intersects $P_i$. The intersection consists of one or more intervals $[t_\mathrm{in}, t_\mathrm{out}] \subset (0,1)$. The ray is inside $B_i$ in the vertical dimension if and only if $z(t^*) \leq h_i$ for some $t^* \in [t_\mathrm{in}, t_\mathrm{out}]$. Since $z(t)$ is linear, its minimum over $[t_\mathrm{in}, t_\mathrm{out}]$ is $\min\!\big(z(t_\mathrm{in}),\, z(t_\mathrm{out})\big)$. The condition $z(t^*) \geq 0$ is always satisfied since $h_\mathrm{tx} > 0$ and $h_\mathrm{rx} > 0$. Therefore, the ray clears $B_i$ vertically if and only if:
\begin{equation}
\min\!\big(z(t_\mathrm{in}),\, z(t_\mathrm{out})\big) > h_i.
\tag{S9}
\end{equation}

\textit{Global LoS condition.}
The ray has a clear LoS if and only if it clears every building. Since the footprints are non-overlapping, the horizontal intersection tests are independent across buildings. The global condition is:
\begin{equation}
\mathrm{LoS}_\mathrm{3D} = \bigwedge_{i=1}^{N} \Big[ \underbrace{\big(r_\mathrm{2D}(t) \notin P_i,\; \forall t \in (0,1)\big)}_{\text{no horizontal intersection}} \;\vee\; \underbrace{\min\!\big(z(t_\mathrm{in}^i),\, z(t_\mathrm{out}^i)\big) > h_i}_{\text{clears rooftop}} \Big].
\tag{S10}
\end{equation}

This is precisely the conjunction of $\mathrm{LoS}_\mathrm{2D\text{-}H}$ and $\mathrm{LoS}_\mathrm{2D\text{-}V}$. The decomposition introduces no approximation because:
\begin{enumerate}
\item Each building $B_i$ is a Cartesian product of a 2D footprint and a 1D height interval, making horizontal and vertical dimensions separable.
\item The ray height $z(t)$ is linear, so its extremum over any interval is attained at an endpoint.
\item Building footprints are non-overlapping, so per-building tests are independent.
\end{enumerate}

Therefore Eq.~(S7) holds exactly under the stated assumptions.
\end{proof}

\begin{remark}[Non-convex footprints]
\label{rem:nonconvex}
When a building footprint $P_i$ is non-convex, the horizontal ray-footprint intersection may consist of multiple disjoint intervals $\{[t_\mathrm{in}^{(k)}, t_\mathrm{out}^{(k)}]\}_{k=1}^{K}$. The vertical check in Eq.~(S9) is then applied independently to each interval. The decomposition remains exact because the linearity of $z(t)$ and the Cartesian product structure hold regardless of footprint convexity.
\end{remark}

\begin{remark}[Limitations of the 2.5D assumption]
\label{rem:2.5d}
The 2.5D prism assumption requires each building to have a uniform height $h_i$. This does not hold for buildings with complex rooftop geometries such as pitched roofs, domes, or rooftop equipment. In such cases, the actual rooftop surface is a function $h_i(x, y)$ rather than a constant, and the vertical check in Eq.~(S9) must be replaced by:
\begin{equation}
\min_{t \in [t_\mathrm{in}, t_\mathrm{out}]} \big[ z(t) - h_i\!\big(r_\mathrm{2D}(t)\big) \big] > 0,
\tag{S11}
\end{equation}
which cannot be evaluated from endpoint values alone when $h_i(x,y)$ is non-linear. The approximation error is bounded by $\max_{(x,y) \in P_i} |h_i(x,y) - \bar{h}_i|$, where $\bar{h}_i$ is the uniform height used in the 2.5D model.
\end{remark}



\section*{Supplementary Results: Beam Pattern Integration}

The main text demonstrates post-hoc beam pattern integration with a 15$^\circ$ half-power beamwidth. Here we provide the complete set of results across four beamwidths (15$^\circ$, 30$^\circ$, 45$^\circ$, 60$^\circ$) and supplementary qualitative comparisons. In all cases, the antenna gain pattern $G_\mathrm{ant}(\phi_\mathrm{AoD}, \theta_\mathrm{AoD})$ is applied to per-ray AoD after the ray tracing computation, without re-executing visibility determination or path search.

\subsection*{15$^\circ$ beamwidth: qualitative APS and PDP}

The main text presents 15$^\circ$ beam RSS maps and quantitative violin distributions. Supplementary Figs.~\ref{fig:supp_aps_15} and \ref{fig:supp_pdp_15} show the corresponding APS and PDP qualitative comparisons. The narrow 15$^\circ$ beam acts as a stringent test of per-ray AoD accuracy: small angular errors in the predicted rays are amplified by the sharp beam roll-off, causing misclassified rays to be either strongly attenuated or incorrectly retained. D$^2$LoS reproduces the beam-filtered APS profiles with correct dominant peak positions and sidelobe suppression. In the PDP domain, the narrow beam removes most off-axis multipath, leaving only a few dominant delay components. D$^2$LoS captures both their positions and relative power levels, whereas RadioUNet and RMTransformer produce delay profiles that bear no resemblance to the ground truth.

\begin{figure}[ht]
  \centering
  \includegraphics[width=\linewidth, height=0.38\textheight, keepaspectratio]{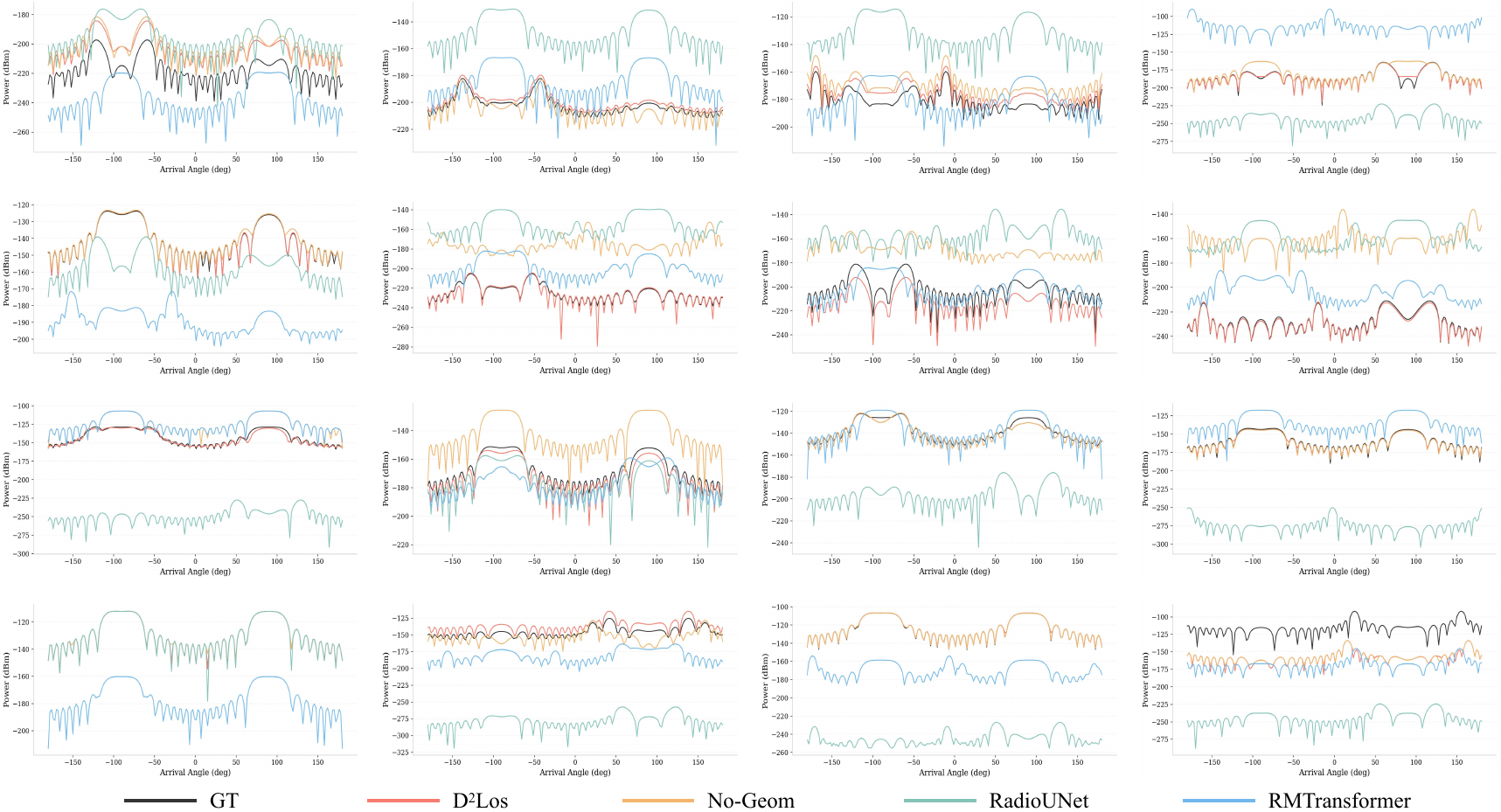}
  \caption{\textbf{APS comparison with a 15$^\circ$ beam pattern.} The directional pattern narrows the angular spread and amplifies the dominant arrival direction. D$^2$LoS reproduces the beam-filtered angular profile. No-Geom preserves the general shape but introduces secondary peak errors. RadioUNet and RMTransformer produce profiles unrelated to the ground truth.}
  \label{fig:supp_aps_15}
\end{figure}

\begin{figure}[ht]
  \centering
  \includegraphics[width=\linewidth, height=0.38\textheight, keepaspectratio]{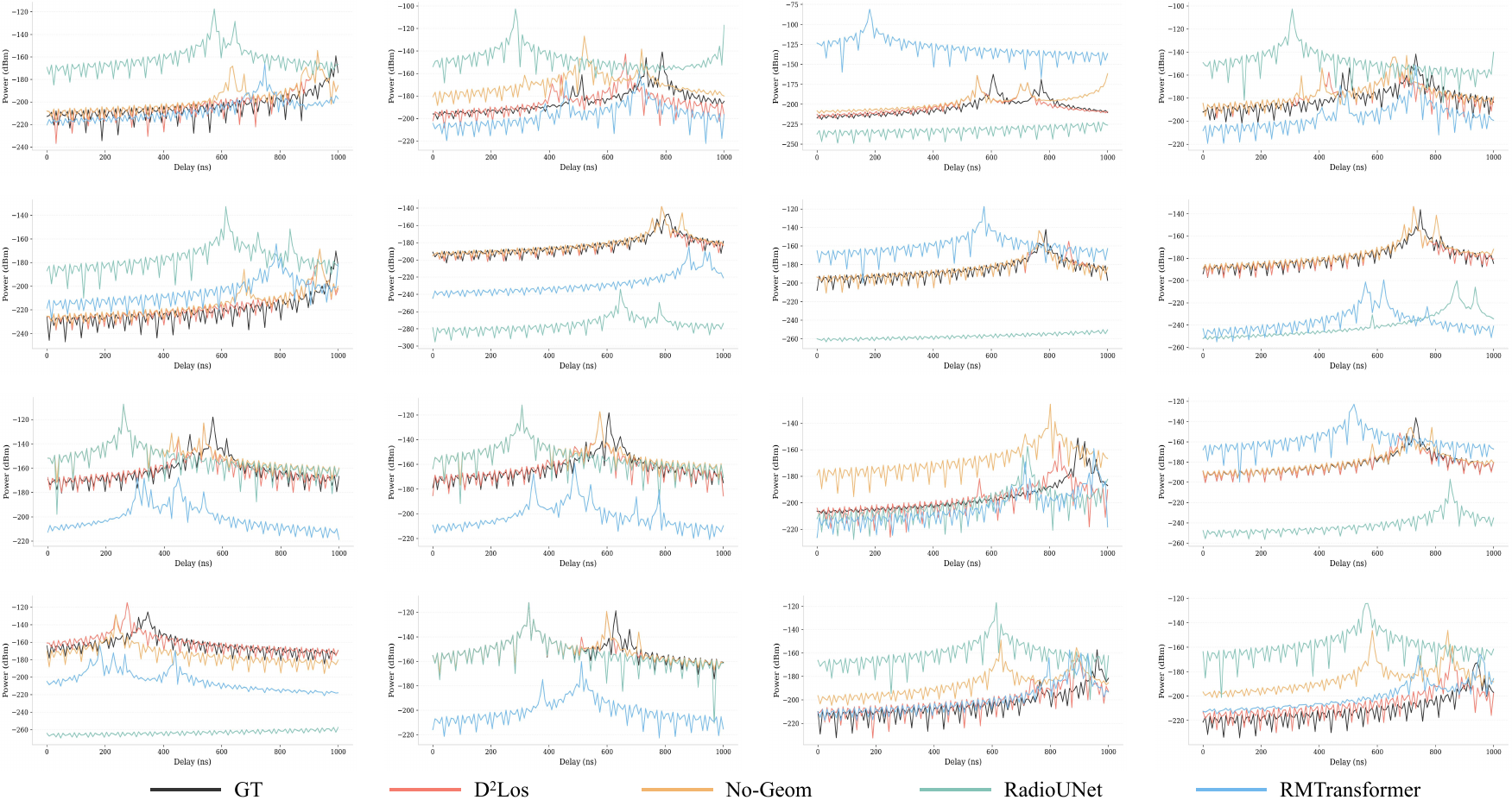}
  \caption{\textbf{PDP comparison with a 15$^\circ$ beam pattern.} The beam pattern suppresses off-axis multipath, reducing the number of resolvable delay components. D$^2$LoS correctly captures this filtering effect. RadioUNet and RMTransformer produce delay profiles with incorrect power levels and spurious peaks.}
  \label{fig:supp_pdp_15}
\end{figure}

\subsection*{30$^\circ$ beamwidth}

Increasing the beamwidth to 30$^\circ$ admits more off-axis multipath into the received signal. This relaxes the angular selectivity compared to the 15$^\circ$ case, making the beam-filtered channel closer to the omnidirectional baseline. As shown in Fig.~\ref{fig:supp_violin_30}, D$^2$LoS maintains compact error distributions across all three channel dimensions. The RSS MAE remains within 1~dB of the omnidirectional result, and the APS shape cosine stays above 0.85. No-Geom shows moderate degradation, while RadioUNet and RMTransformer continue to produce heavy-tailed error distributions. The qualitative RSS maps (Fig.~\ref{fig:supp_rss_30}) confirm that D$^2$LoS preserves the wider beam footprint and shadow boundaries. The APS and PDP curves (Figs.~\ref{fig:supp_aps_30} and \ref{fig:supp_pdp_30}) show that D$^2$LoS tracks both the dominant and secondary multipath components under the 30$^\circ$ beam.

\begin{figure}[ht]
  \centering
  \includegraphics[width=0.48\linewidth]{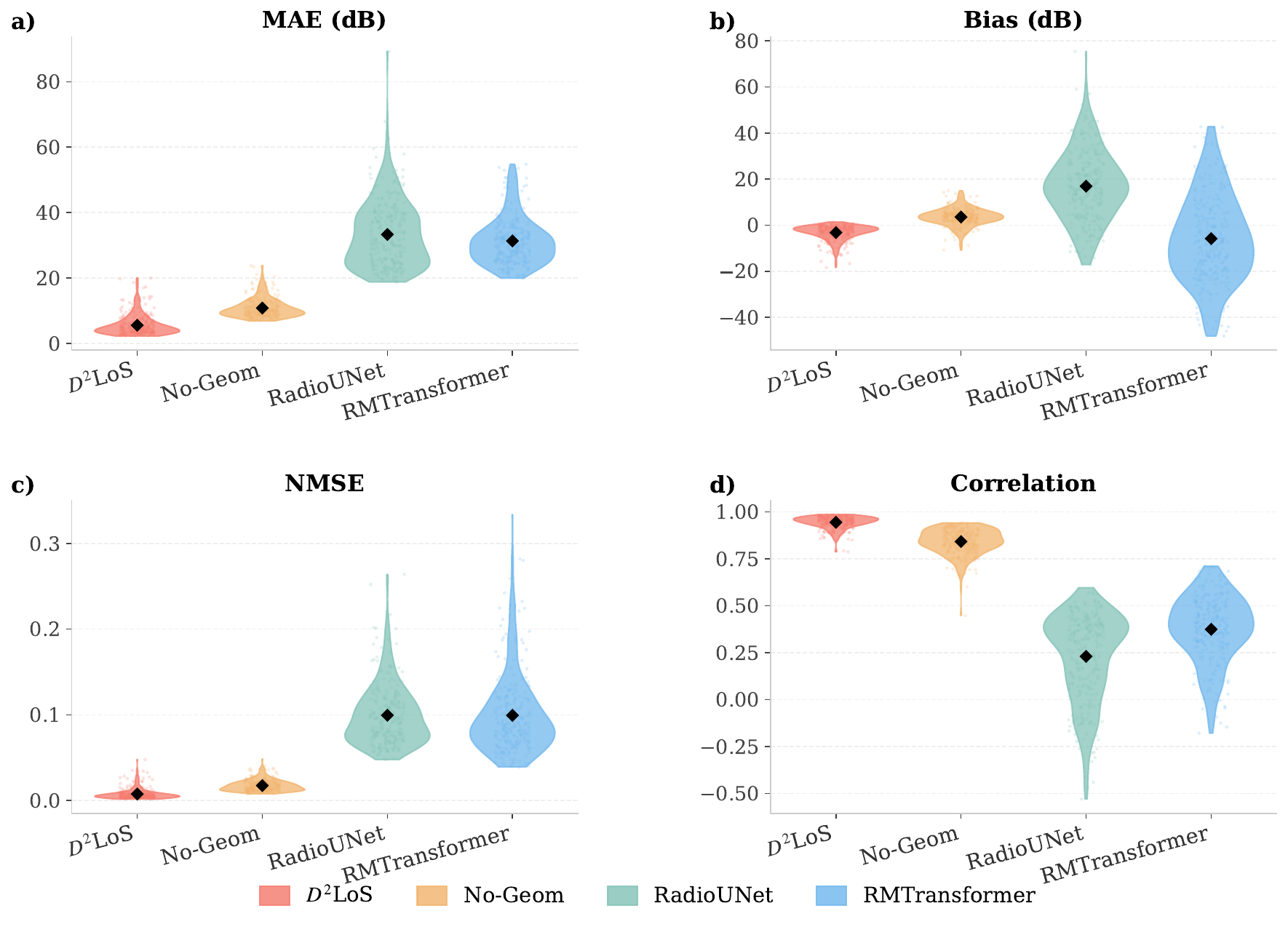}
  \hfill
  \includegraphics[width=0.48\linewidth]{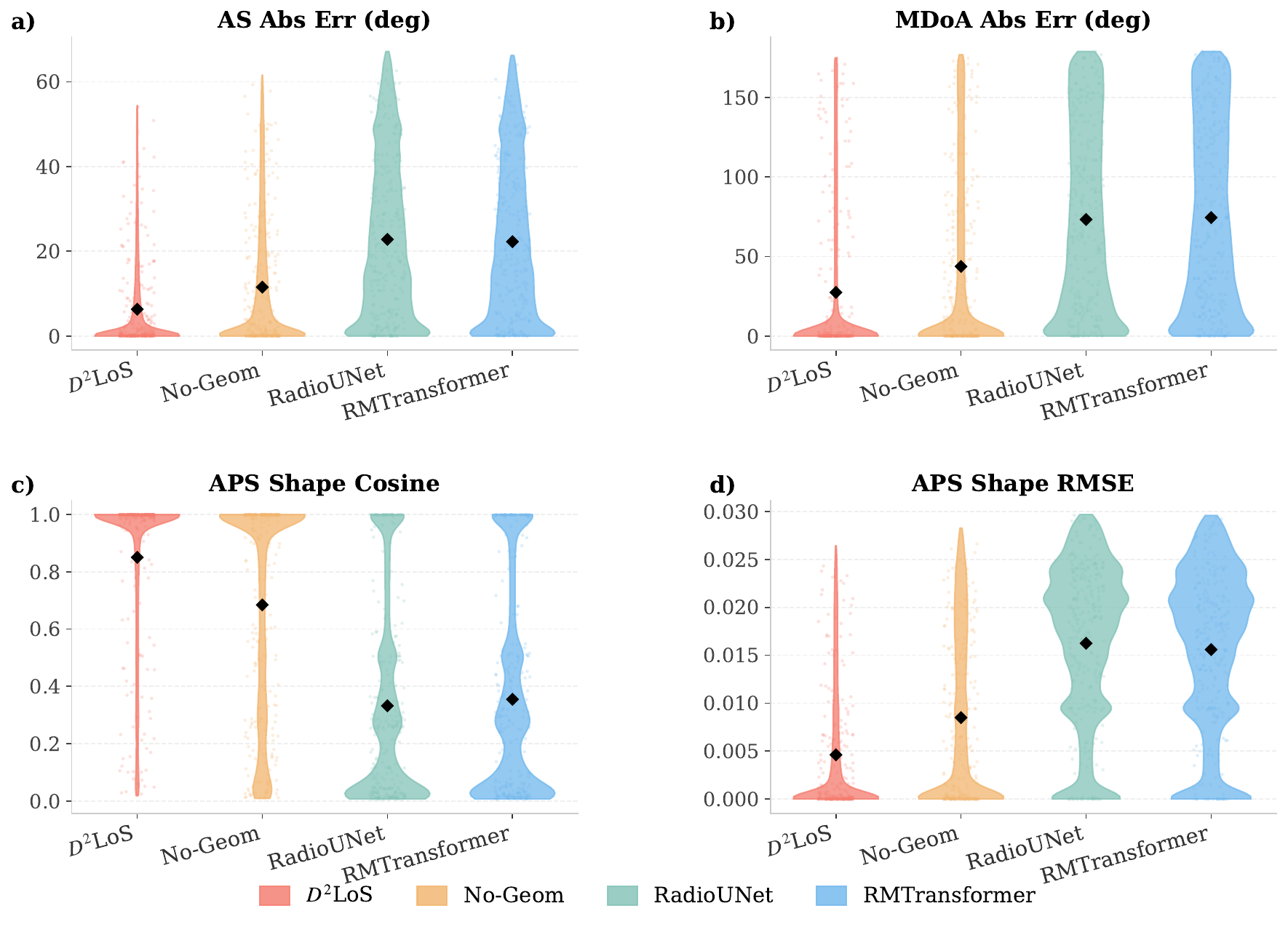}\\[6pt]
  \centering
  \includegraphics[width=0.48\linewidth]{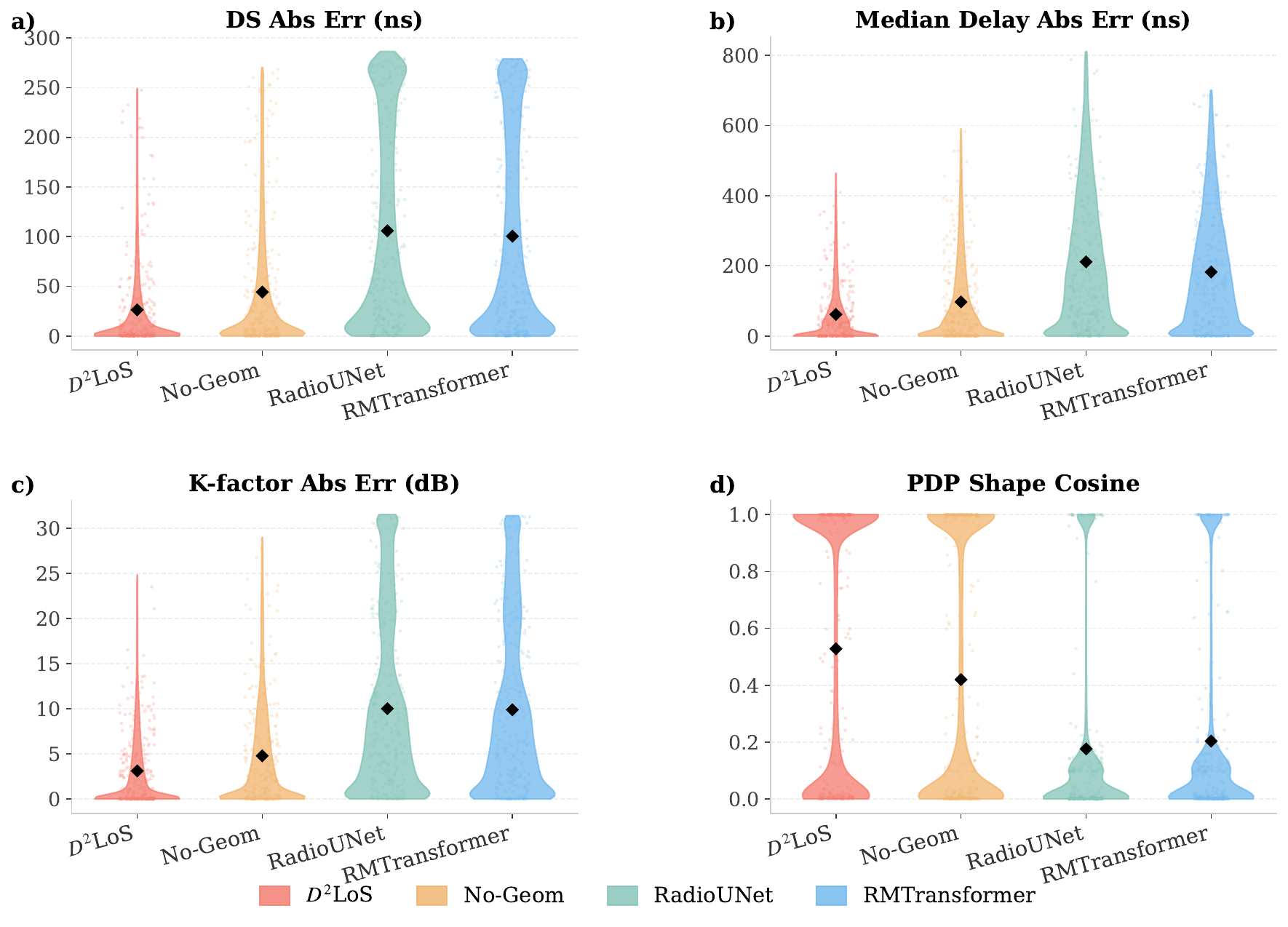}
  \caption{\textbf{Prediction accuracy with a 30$^\circ$ beam pattern.} D$^2$LoS maintains compact, low-error distributions across RSS, APS, and PDP metrics, consistent with the 15$^\circ$ case.}
  \label{fig:supp_violin_30}
\end{figure}

\begin{figure}[ht]
  \centering
  \includegraphics[width=\linewidth, height=0.38\textheight, keepaspectratio]{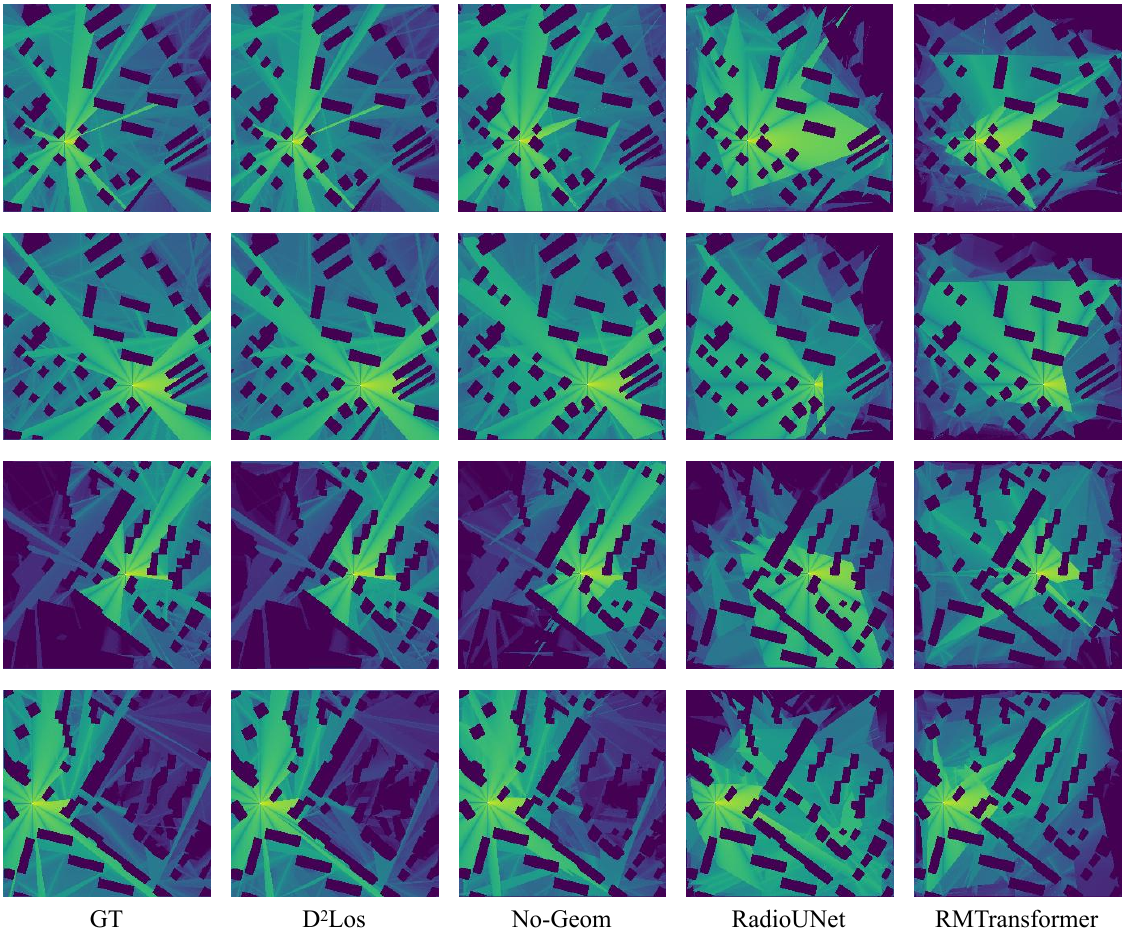}
  \caption{\textbf{RSS radio maps with a 30$^\circ$ beam pattern.} The wider beam produces broader spatial coverage compared to the 15$^\circ$ case. D$^2$LoS preserves both the beam footprint and shadow structure.}
  \label{fig:supp_rss_30}
\end{figure}

\begin{figure}[ht]
  \centering
  \includegraphics[width=\linewidth, height=0.38\textheight, keepaspectratio]{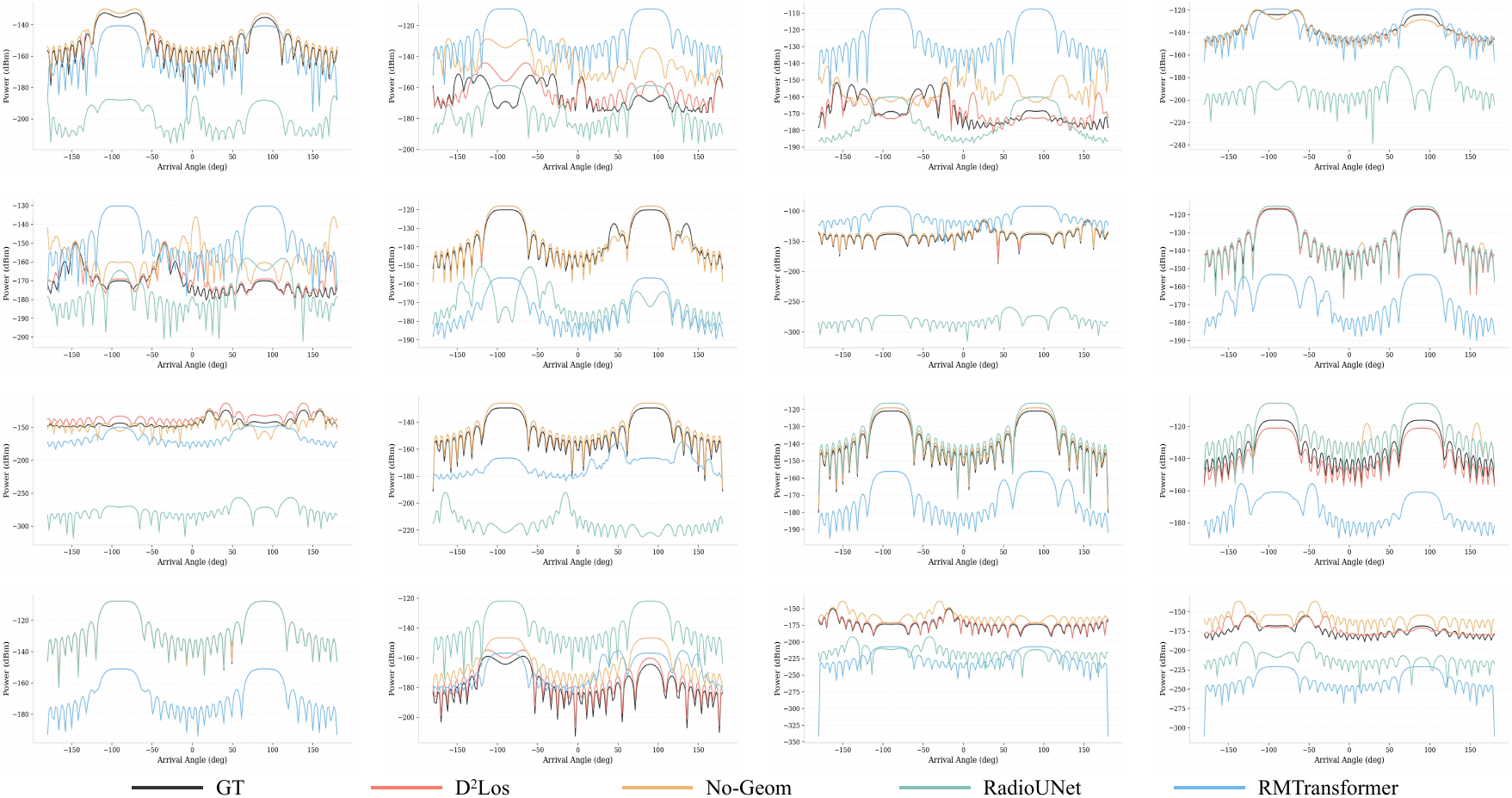}
  \caption{\textbf{APS comparison with a 30$^\circ$ beam pattern.} D$^2$LoS closely tracks the ground truth angular profile under the wider beam.}
  \label{fig:supp_aps_30}
\end{figure}

\begin{figure}[ht]
  \centering
  \includegraphics[width=\linewidth, height=0.38\textheight, keepaspectratio]{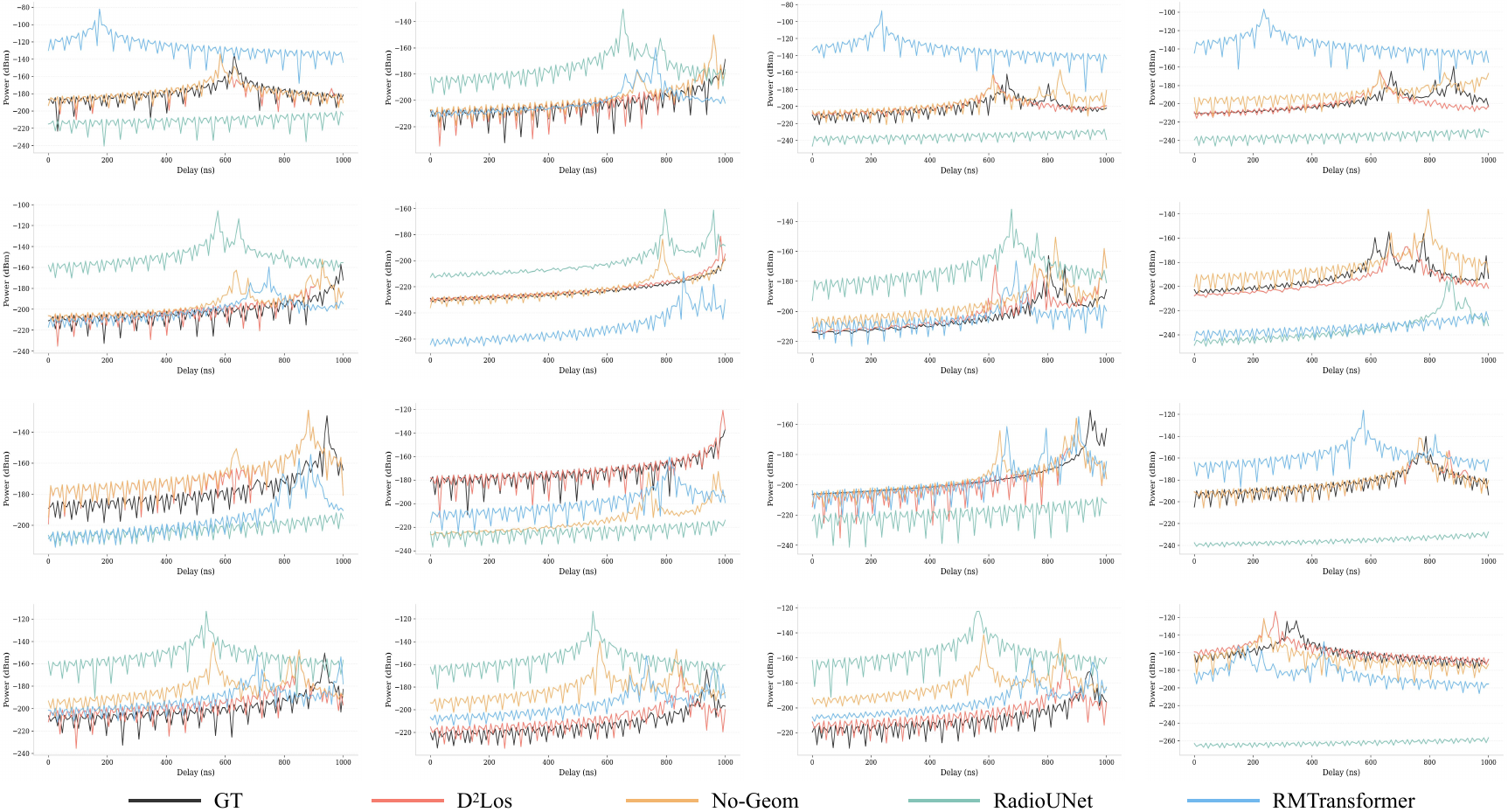}
  \caption{\textbf{PDP comparison with a 30$^\circ$ beam pattern.} The wider beam retains more multipath components than the 15$^\circ$ case. D$^2$LoS reproduces the delay structure accurately.}
  \label{fig:supp_pdp_30}
\end{figure}

\subsection*{45$^\circ$ beamwidth}

At 45$^\circ$ beamwidth, the directional filtering becomes mild and the beam-specific channel approaches the omnidirectional case. Fig.~\ref{fig:supp_violin_45} shows that the error distributions for D$^2$LoS are nearly identical to the omnidirectional violin plots in the main text. This is expected: when the beam is wide enough to capture most multipath components, angular errors in individual rays have limited impact on the aggregated metrics. The RSS maps (Fig.~\ref{fig:supp_rss_45}) show only subtle differences from the omnidirectional maps, primarily at the beam edges where slight power attenuation is visible. The APS and PDP curves (Figs.~\ref{fig:supp_aps_45} and \ref{fig:supp_pdp_45}) closely match the omnidirectional results for D$^2$LoS, while RadioUNet and RMTransformer remain unable to produce physically meaningful outputs regardless of beamwidth.

\begin{figure}[ht]
  \centering
  \includegraphics[width=0.48\linewidth]{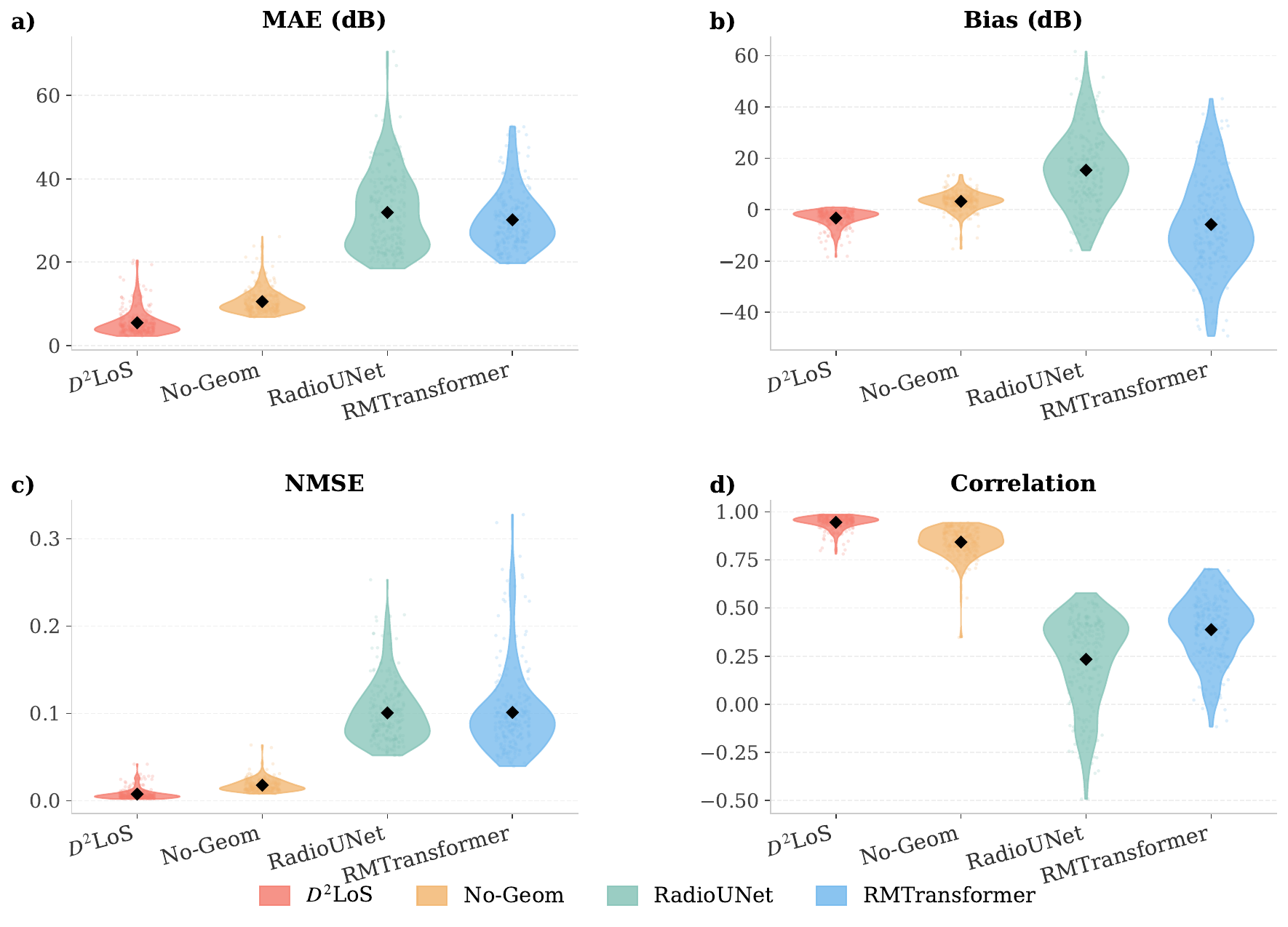}
  \hfill
  \includegraphics[width=0.48\linewidth]{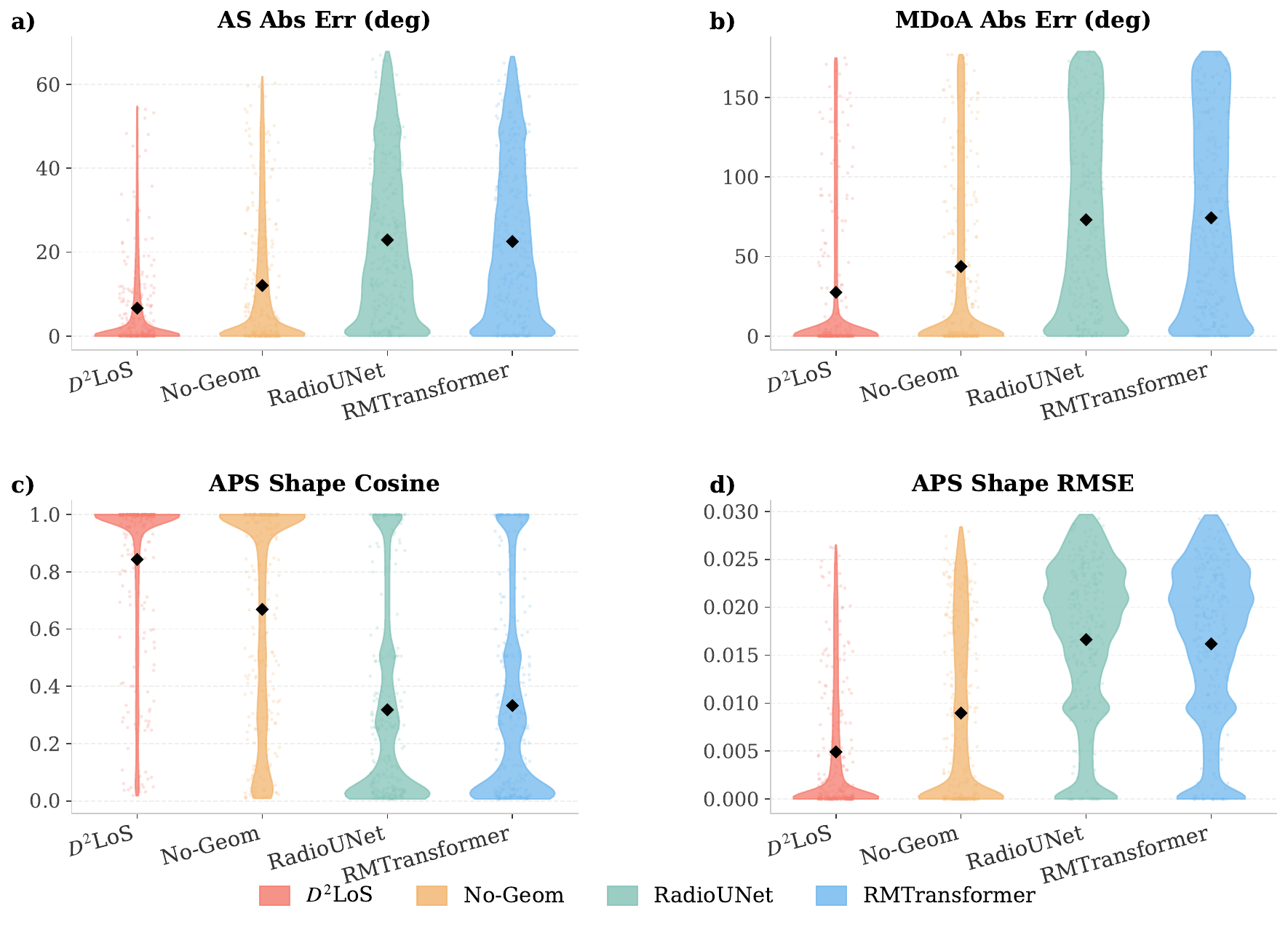}\\[6pt]
  \centering
  \includegraphics[width=0.48\linewidth]{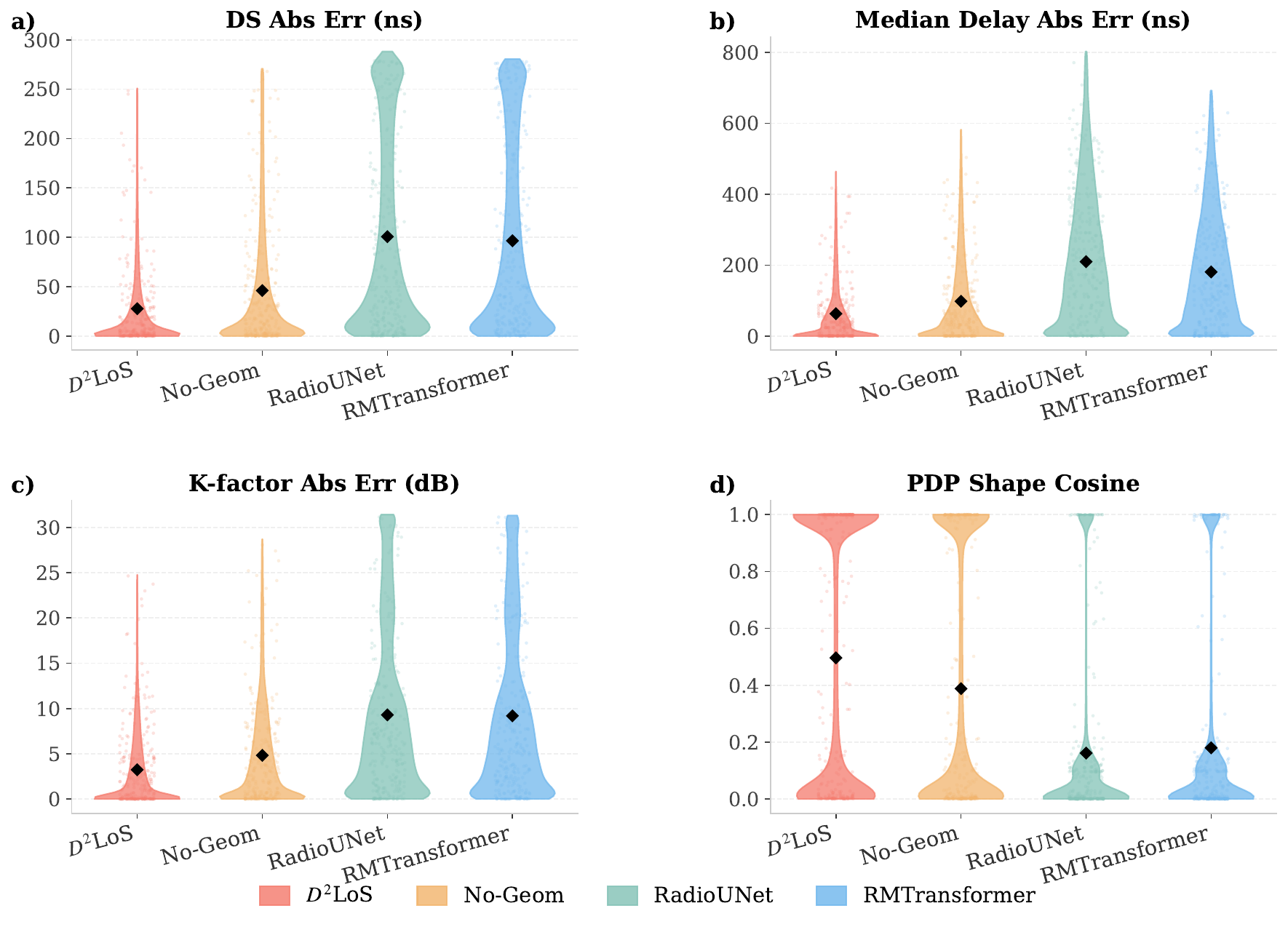}
  \caption{\textbf{Prediction accuracy with a 45$^\circ$ beam pattern.} Error distributions are nearly identical to the omnidirectional results, confirming that the wide beam captures most multipath energy.}
  \label{fig:supp_violin_45}
\end{figure}

\begin{figure}[ht]
  \centering
  \includegraphics[width=\linewidth, height=0.38\textheight, keepaspectratio]{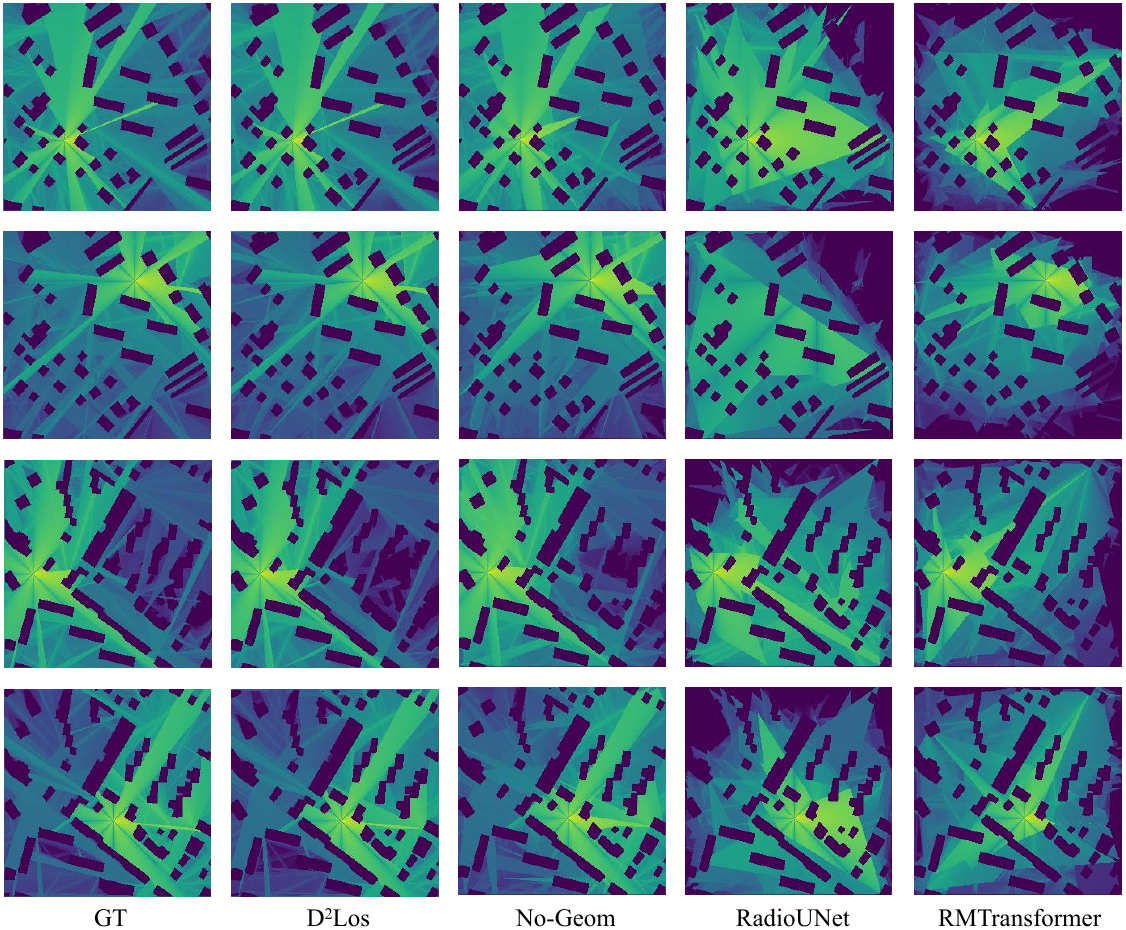}
  \caption{\textbf{RSS radio maps with a 45$^\circ$ beam pattern.} The coverage pattern is close to omnidirectional with mild directional shaping at the beam edges.}
  \label{fig:supp_rss_45}
\end{figure}

\begin{figure}[ht]
  \centering
  \includegraphics[width=\linewidth, height=0.38\textheight, keepaspectratio]{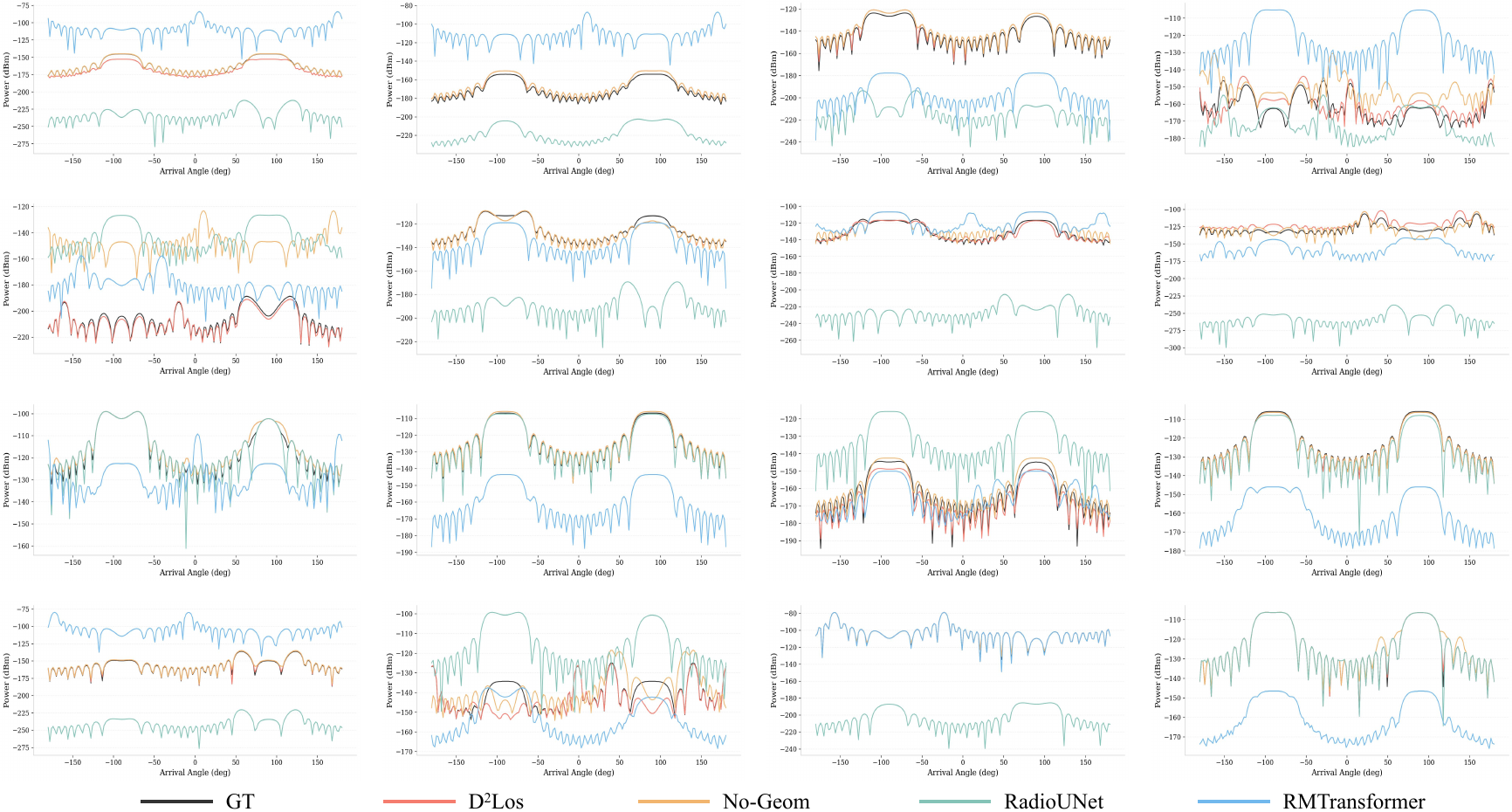}
  \caption{\textbf{APS comparison with a 45$^\circ$ beam pattern.} D$^2$LoS closely matches the ground truth. The wider beam preserves most angular features from the omnidirectional case.}
  \label{fig:supp_aps_45}
\end{figure}

\begin{figure}[ht]
  \centering
  \includegraphics[width=\linewidth, height=0.38\textheight, keepaspectratio]{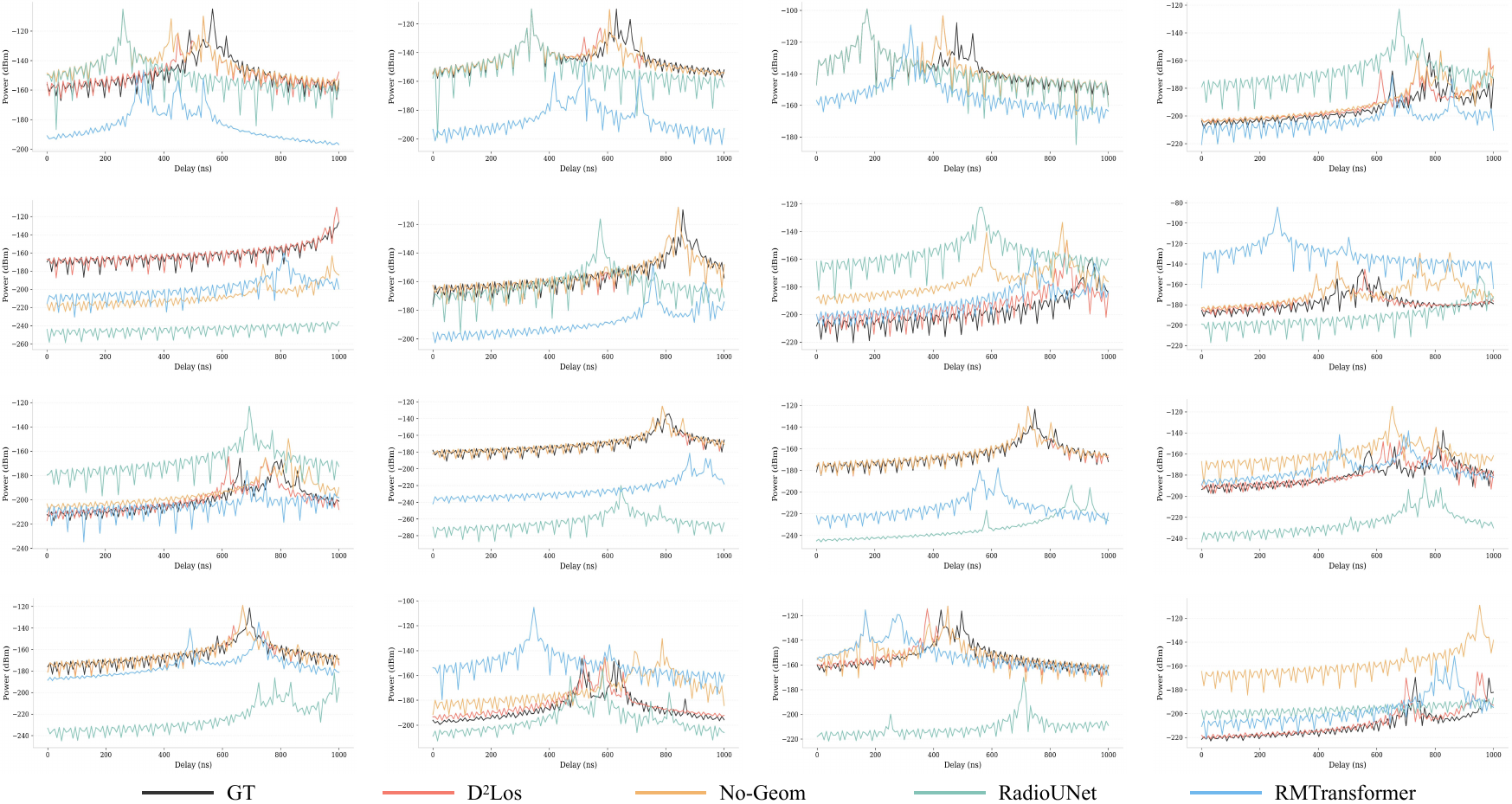}
  \caption{\textbf{PDP comparison with a 45$^\circ$ beam pattern.} The delay profile is close to the omnidirectional case. D$^2$LoS maintains accurate reproduction of multipath structure.}
  \label{fig:supp_pdp_45}
\end{figure}

\subsection*{60$^\circ$ beamwidth}

The 60$^\circ$ beamwidth represents the widest tested configuration and the mildest directional filtering. As shown in Fig.~\ref{fig:supp_violin_60}, D$^2$LoS retains low-error, compact distributions that are effectively indistinguishable from the omnidirectional case. The RSS maps (Fig.~\ref{fig:supp_rss_60}) differ from the omnidirectional maps primarily in a slight overall power offset due to the finite antenna gain. The APS and PDP curves (Figs.~\ref{fig:supp_aps_60} and \ref{fig:supp_pdp_60}) confirm that D$^2$LoS accuracy is independent of beamwidth. Taken together with the 15$^\circ$ results, these findings demonstrate that the per-ray angular information produced by D$^2$LoS is sufficiently accurate to support beam-specific analysis from narrow pencil beams to near-omnidirectional patterns, all from a single ray tracing computation.

\begin{figure}[ht]
  \centering
  \includegraphics[width=0.48\linewidth]{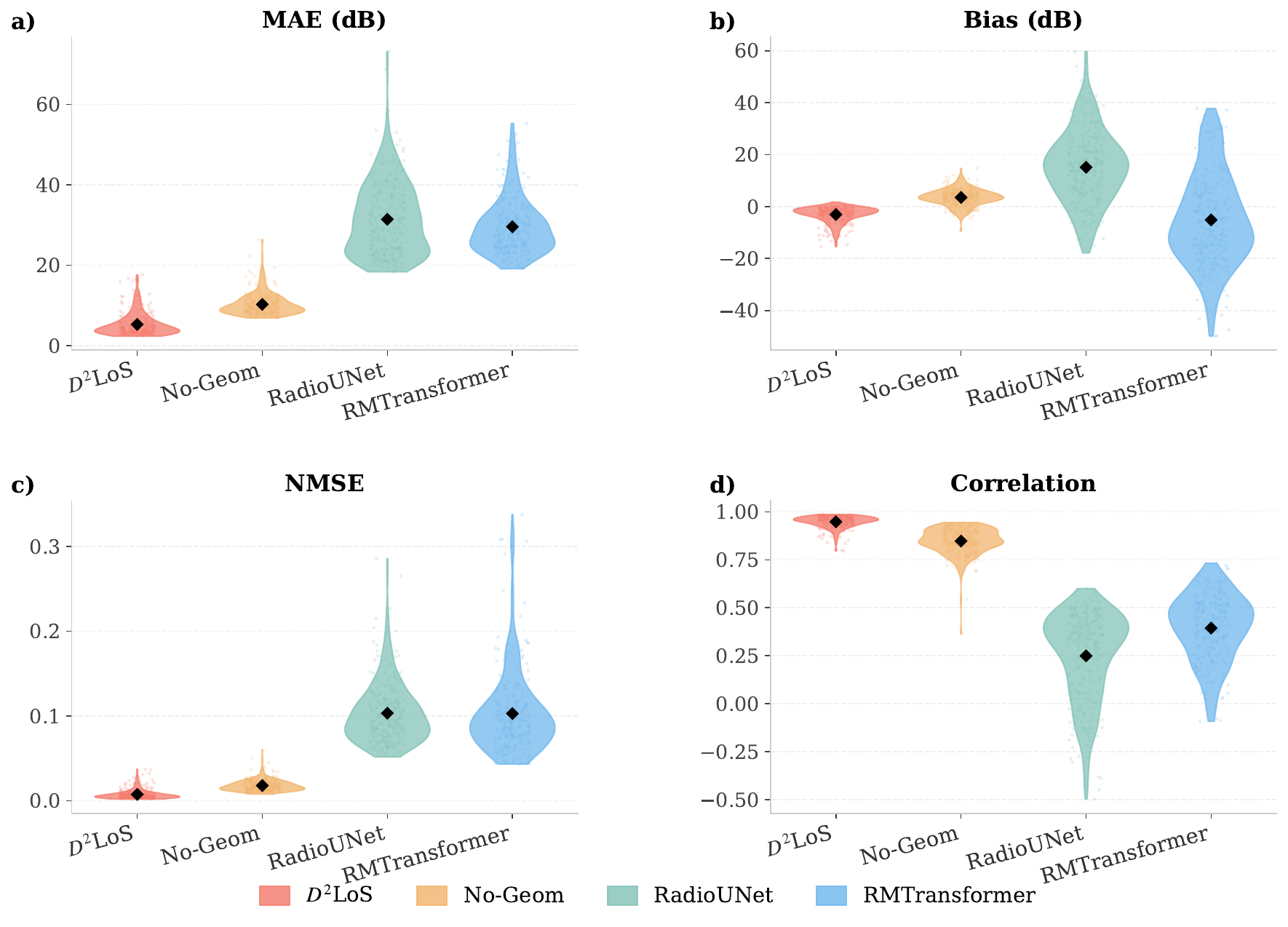}
  \hfill
  \includegraphics[width=0.48\linewidth]{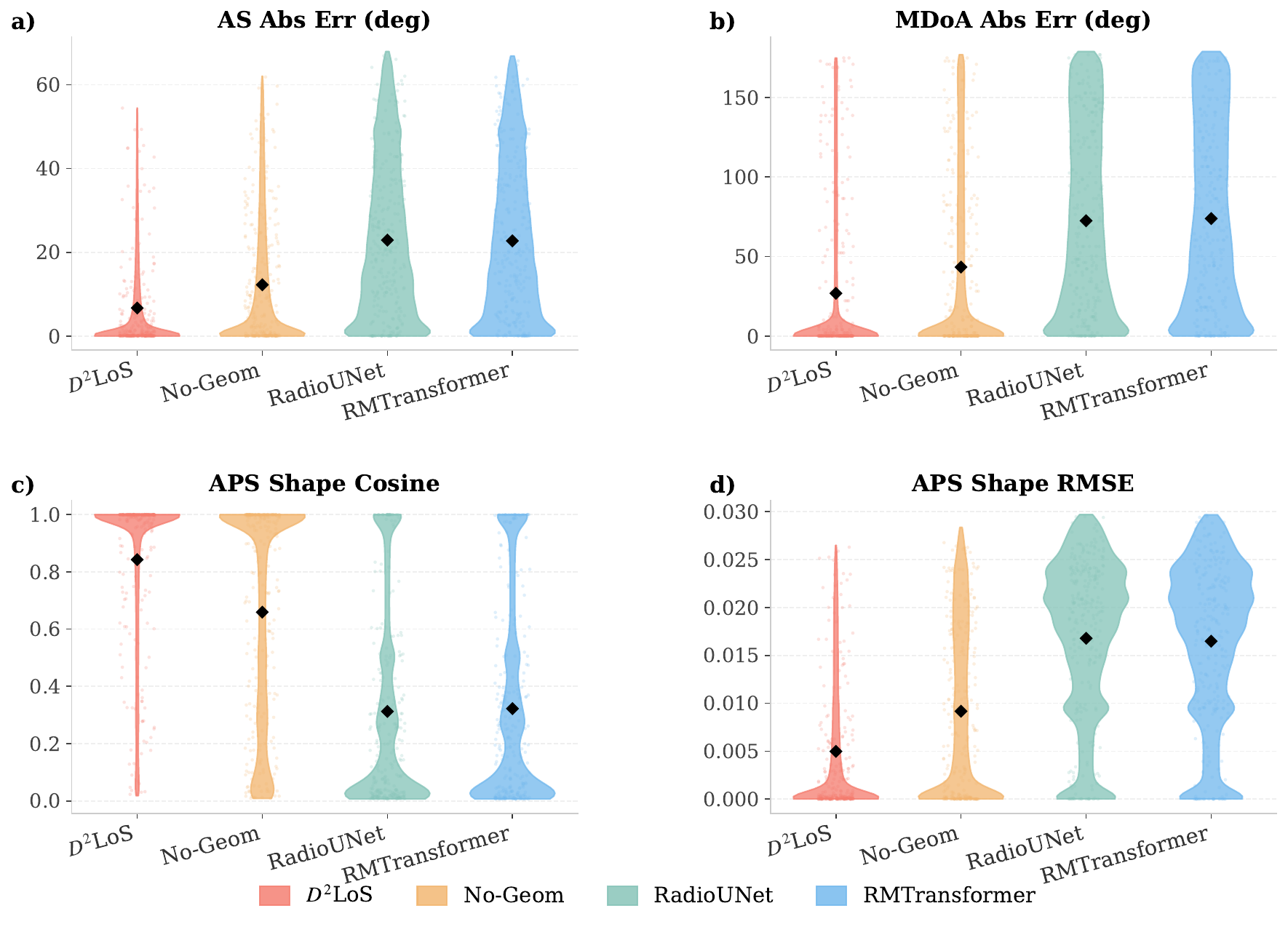}\\[6pt]
  \centering
  \includegraphics[width=0.48\linewidth]{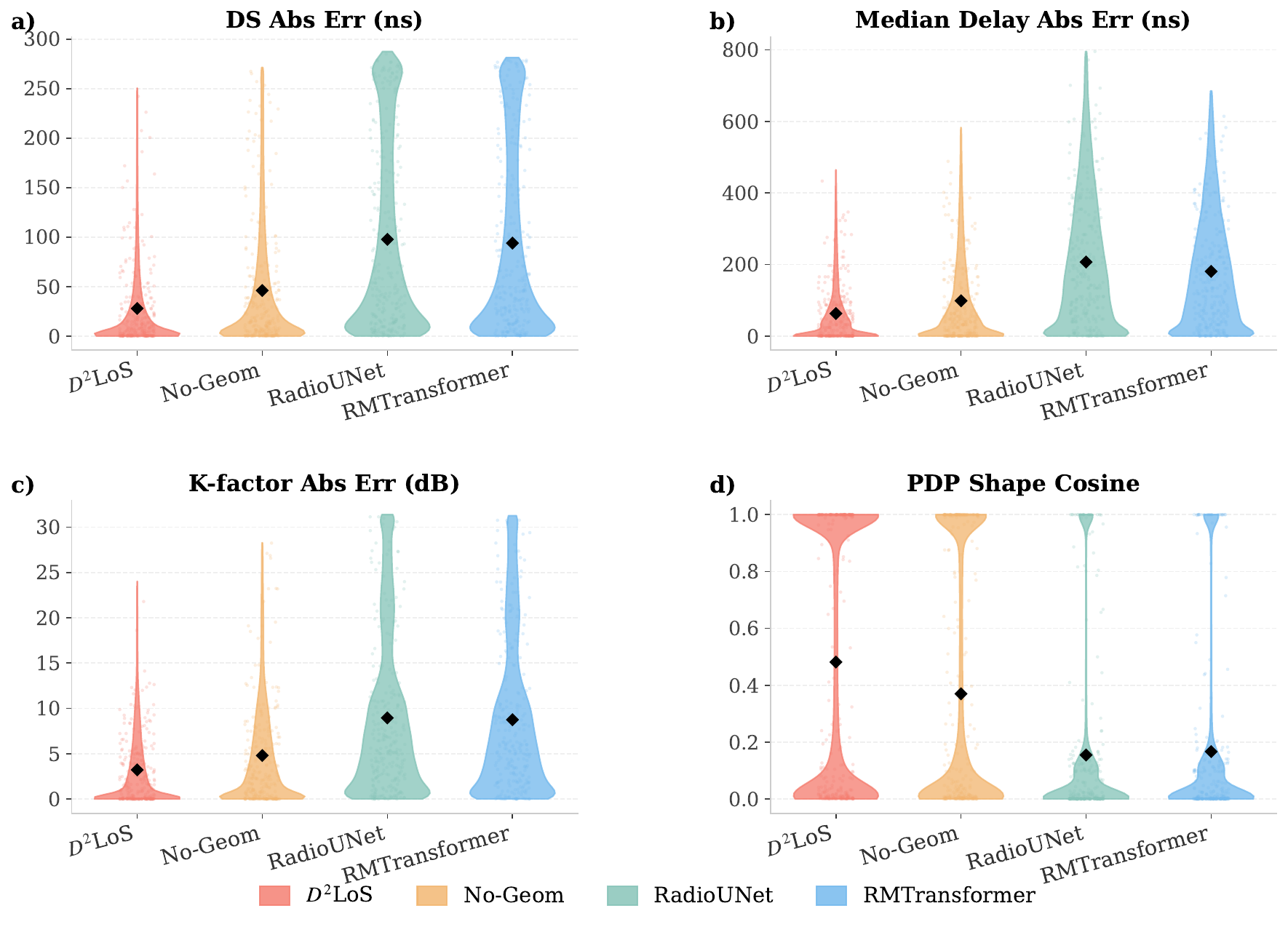}
  \caption{\textbf{Prediction accuracy with a 60$^\circ$ beam pattern.} Error distributions are effectively indistinguishable from the omnidirectional case, confirming robustness across all tested beamwidths.}
  \label{fig:supp_violin_60}
\end{figure}

\begin{figure}[ht]
  \centering
  \includegraphics[width=\linewidth, height=0.38\textheight, keepaspectratio]{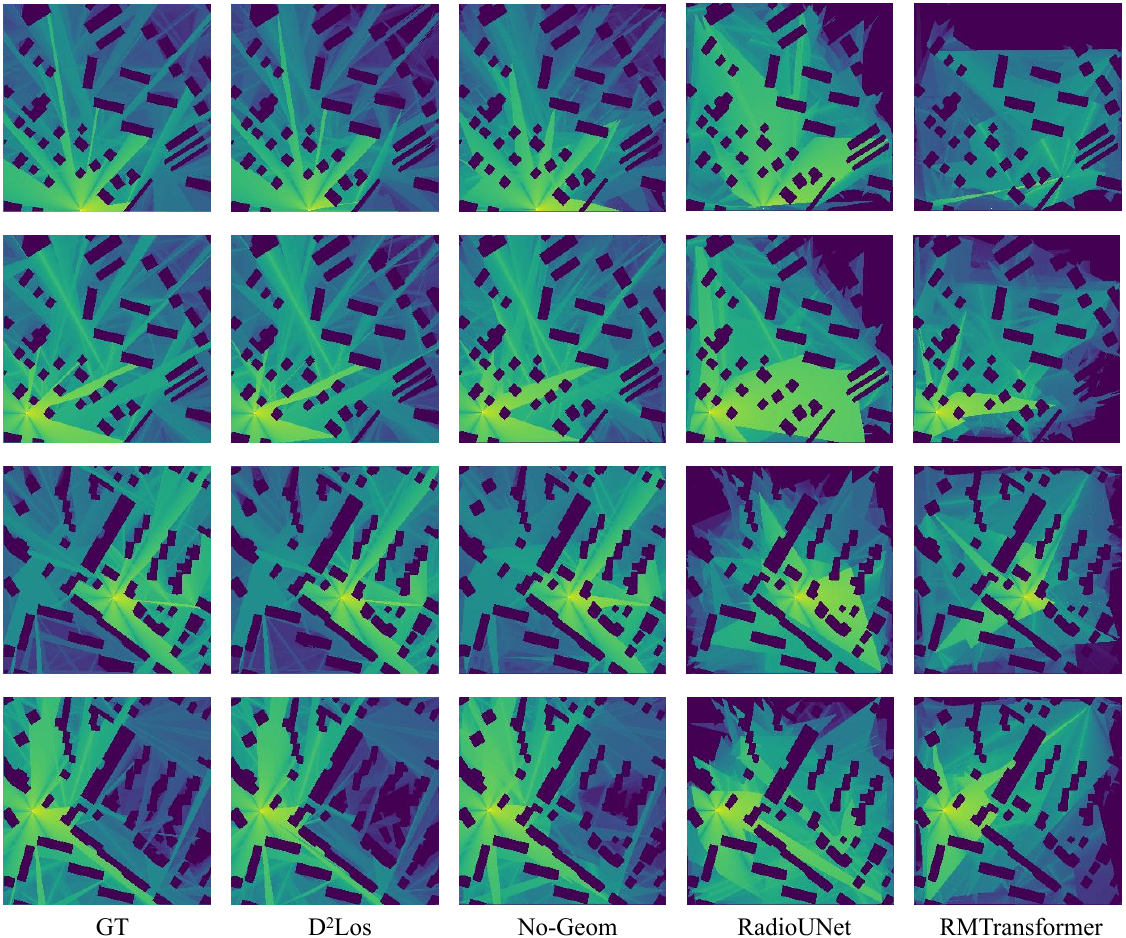}
  \caption{\textbf{RSS radio maps with a 60$^\circ$ beam pattern.} D$^2$LoS accurately reproduces the near-omnidirectional coverage with mild directional shaping.}
  \label{fig:supp_rss_60}
\end{figure}

\begin{figure}[ht]
  \centering
  \includegraphics[width=\linewidth, height=0.38\textheight, keepaspectratio]{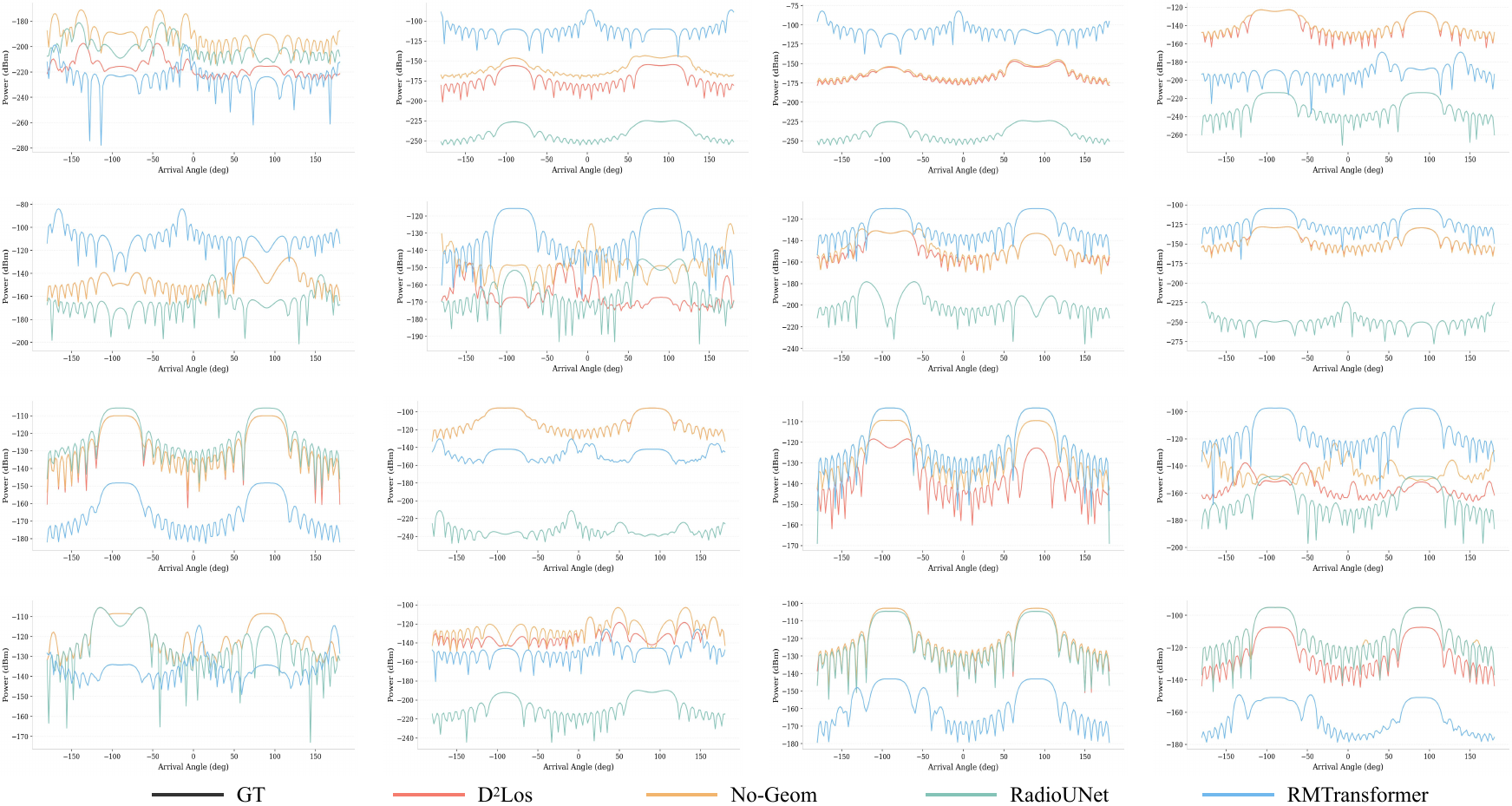}
  \caption{\textbf{APS comparison with a 60$^\circ$ beam pattern.} The angular profile is close to the omnidirectional case. D$^2$LoS maintains full accuracy.}
  \label{fig:supp_aps_60}
\end{figure}

\begin{figure}[ht]
  \centering
  \includegraphics[width=\linewidth, height=0.38\textheight, keepaspectratio]{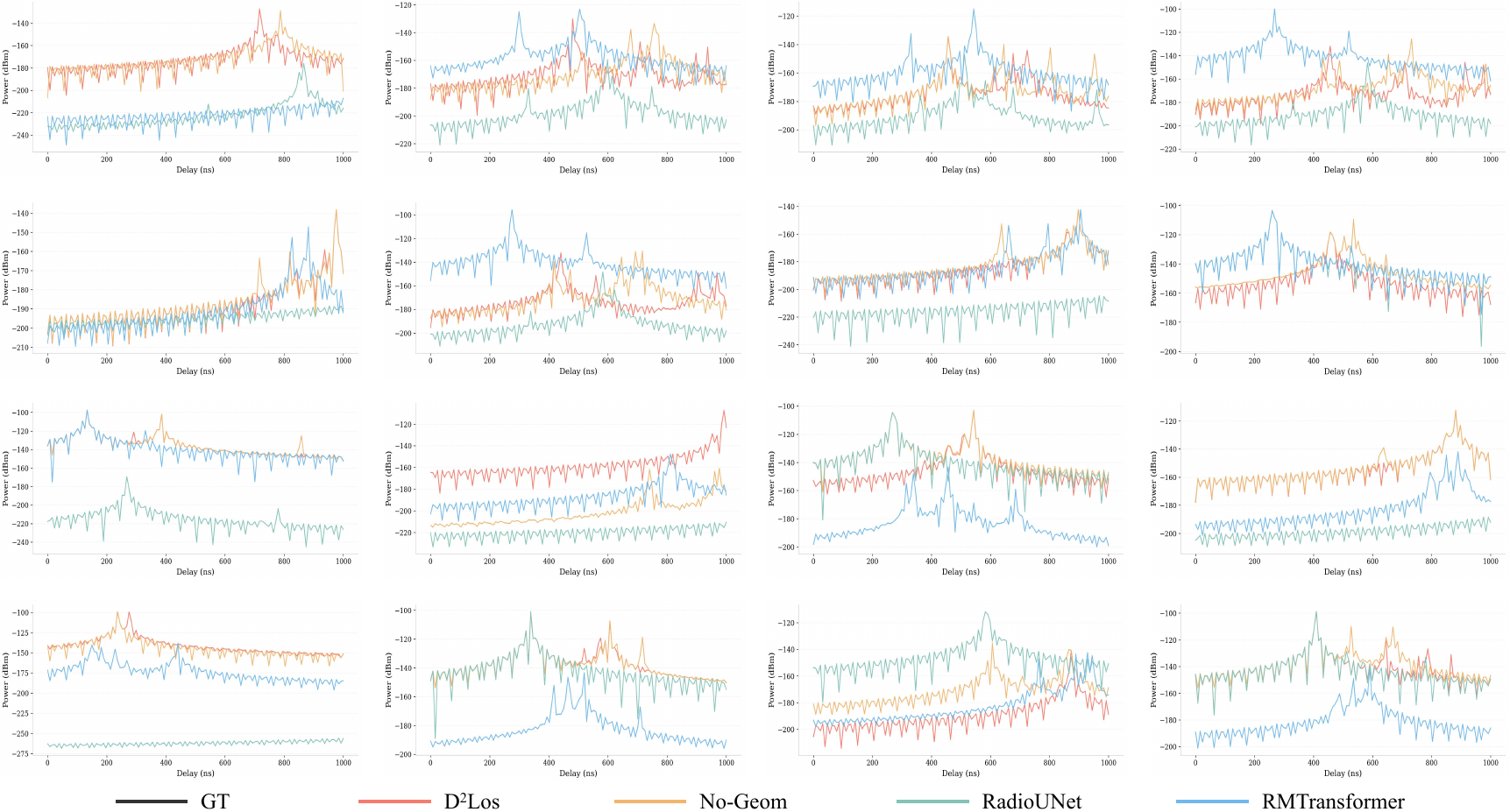}
  \caption{\textbf{PDP comparison with a 60$^\circ$ beam pattern.} Delay structure matches the omnidirectional case closely. D$^2$LoS reproduces it accurately.}
  \label{fig:supp_pdp_60}
\end{figure}
\end{appendices}

\section{Data availability}
Data is available at \url{https://github.com/UNIC-Lab/RayVerse}.

\section{Code availability}
Data is available at \url{https://github.com/UNIC-Lab/D2LoS}.

\bibliography{sn-bibliography}

\end{document}